\documentclass[journal,twocolumn,10pt]{IEEEtran}

\usepackage{latexsym,amsfonts,amsmath,url,bm} 
\usepackage{boxedminipage,graphicx,color,array,multicol,ulem,subfigure}
\usepackage[abs]{overpic}
\usepackage{amsthm}
\usepackage[utf8]{inputenc}
\usepackage[english]{babel}
\usepackage{fancyhdr}

\usepackage[dvipsnames]{xcolor}

\DeclareMathOperator*{\argmin}{arg\,min}  
\DeclareMathOperator*{\argmax}{arg\,max}

\newcommand{\E}{\mathrm{E}}  
 
\bgroup\catcode`\#=12\egroup \newcommand{\eqdef}{\stackrel{\textrm{\tiny def}}{=}}
\newcommand{\myctot}{C}

\newtheorem{proposition}{Proposition}

\begin{document}

\title{On the estimation of spatial density \\ from mobile network operator data}

\author{\IEEEauthorblockN{Fabio Ricciato and Angelo Coluccia} 
\thanks{F. Ricciato is with European Commission, DG Eurostat, Unit A5 Methodology; Innovation in Official Statistics, Luxembourg. Email: fabio.ricciato@ec.europa.eu.}
\thanks{A. Coluccia is with University of Salento, Lecce, Italy. Email: angelo.coluccia@unisalento.it.}
\thanks{The views expressed in this paper are those of the authors and do not necessarily represent the official position of the European Commission.}
}

\maketitle

\begin{abstract}
We tackle the problem of estimating the spatial distribution of mobile phones from Mobile Network Operator (MNO) data, namely Call Detail Record (CDR) or signalling data. 
The process of transforming MNO data to a density map 
requires geolocating radio cells to determine
their spatial footprint. Traditional geolocation solutions rely on Voronoi tessellations and approximate cell footprints by mutually disjoint regions.
Recently, some pioneering work started to consider more elaborate geolocation methods with partially overlapping (non-disjoint) cell footprints coupled with a probabilistic model for phone-to-cell association. Estimating the spatial density in such a probabilistic setup is currently an open research problem and is the focus of the present work. We start by reviewing three different estimation methods proposed in literature and provide novel analytical insights that unveil some key aspects of their mutual relationships and properties. Furthermore, we develop a novel estimation approach for which a closed-form solution can be given. Numerical results based on semi-synthetic data are presented to assess the relative accuracy of each method. Our results indicate that the estimators based on overlapping cells have the potential to improve spatial accuracy over traditional approaches based on Voronoi tessellations. 
\end{abstract}


\section{Introduction}
Most people nowadays carry a mobile phone. Mobile phones interact several times a day with the mobile network infrastructure, and every interaction reveals the approximate location of the phone to the network, at least at radio cell level. Such interactions are recorded by the Mobile Network Operator (MNO) for purposes related to the delivery of mobile communication services, e.g., billing, network optimization and troubleshooting. 
The class of MNO data includes the more common Call Detailed Records (CDR), fetched from the billing system, as well as the more informative signaling data, acquired by dedicated network monitoring systems \cite{metawin06}.

Since more than two decades, researchers have (re)used MNO data, most prominently CDR but increasingly also signaling data,
to study patterns of human mobility and presence, effectively exploiting the mobile network infrastructure  as a large-scale ``sensor of opportunity" for human mobility flows (see \cite{fiore2015survey} and references therein). More recently, MNO data have been used to analyze human mobility in the context of Covid-19 pandemic and associated containment measures \cite{JRCcovid1,JRCcovid2,covid2021} notwithstanding a number of open challenges in accessing and making use of such data \cite{nuria2020covid}. 

Beyond the scope of academic research,  companies and non-governmental organizations are offering insights and statistical products derived from MNO data in various application domains, from humanitarian support to analysis of tourism flows, and statistical organizations are looking with increasing interest at MNO data as a potential source for compiling new official statistics \cite{unsd2014,meersman2016,essnetWPIdeliverableI3}. 
 However, despite the great volume of research literature on the topic, several methodological aspects remain open along the journey from raw MNO data to reliable  summary information. 

The focus of this contribution is the problem of measuring the spatial density of mobile phones at a given reference time based on the data from a single MNO collected.  The density of mobile phones provides a rough proxy of the spatial distribution of humans in a given territory around that time, also called ``present population", ``hourly population" or ``\emph{de facto} population''.

The data processing flow, from raw MNO data to final density estimates, involves modeling the spatial coverage patterns of radio cells or groups thereof, i.e., mapping each radio cell to a geographical territory. This logical function,  called ``cell geolocation'' in earlier literature \cite{ricciatoworkingpaper,essnetWPIdeliverableI3},
can be performed in different ways with varying levels of sophistication as to what kinds of data are taken into consideration  and how the radio propagation phenomenon is modeled.   
Choosing a more  elaborate method accounts to attempting a more detailed and complex modeling of  the mobile network infrastructure in order to potentially, but not necessarily, achieve better accuracy.

Geolocation methods can be grouped into two classes, discriminated by whether the set of radio cell locations are modeled as spatially disjoint or, conversely, overlapping between cell locations is allowed. The methods involving disjoint (non-overlapping) cell locations, or ``tessellations", are simpler to implement and by far more popular in existing literature. 
Only recently, a few pioneering papers such as \cite{jrcstudy,ricciato16,moblocworkingpaper,salgado21} have started to consider alternative, more elaborate schemes based on overlapping cells. 

A separate module along the data workflow, logically subsequent to the geolocation module, performs the task of estimating (or inferring) the underlying spatial density from the set of geolocalized data. There is a fundamental inter-dependency between these two modules: if the geolocalization method of choice belongs to the specific class of tessellations, then the estimation task reduces to a simple area-proportional solution. Conversely, if the  geolocalization method belongs to the more general class of overlapping locations, determining the ``best" estimation method is an open research problem, that is the focus of the present  paper. 
Given this background, we make here the following contributions:
\begin{itemize}
\item We provide a coherent representation of the general processing workflow, from individual MNO records to spatial density estimates, with a clear definition of the geolocation and estimation modules. 
\item Focusing on estimation methods for overlapping cells,  we review three distinct solutions from previous literature and provide novel  analytical insights about their properties and mutual relationships. 
\item  We derive a novel ad hoc estimator with a closed-form solution.
\item We provide quantitative insights by comparing the accuracy of the different solutions on a semi-synthetic scenario.
\end{itemize}
Along the way, we identify directions for further research. 

The rest of this paper is organized as follows. 
We start by providing  a general view of the data transformation flow in Section \ref{sec:tutorialMNO}, from the input data sources to the desired output information,  and  formulate  the density estimation problem analytically. In Section \ref{sec:soa}, we review a number of  estimation solutions that were proposed earlier in the literature, for which  we present novel analytical results that re-interpret these methods and provide further insights about their mutual relationships, in Section \ref{sec:newresults}. 
In Section \ref{sec:map}, we propose a completely novel \emph{ad hoc}  estimator for the problem at hand.  
Illustrative numerical results based on (semi-)synthetic data are presented in Section \ref{sec:numres}. Finally, we conclude the paper in Section \ref{sec:conclusions}.

\section{Problem Formulation}\label{sec:tutorialMNO}

\subsection{Desired output}
We address the problem of estimating the spatial density  of mobile phones in a given territory around a reference time $t^*$ starting from the data collected from the network of a single mobile network operator. 
The mobile phone density obtained from single MNO data provides a rough proxy for the human population density, in the sense that the variations in space and time of two  distributions can be reasonably expected to be strongly correlated.  Therefore, the estimated  mobile  phone density can be used directly in those application where 
the spatio-temporal variations are of interest, i.e., where and when density increases or decreases. 
Furthermore,  mobile phone density estimation provides a basis  for the accurate estimation of human population density if additional calibration data are available\footnote{The accurate estimation of the \emph{absolute value} of human population density from mobile phone density is a non-trivial task since humans do not map one-to-one to mobile phones: some persons carry no phone, other persons carry multiple phones, and some phones are not carried by a single person or by no person at all. Calibration methods clearly depend on the kind of available calibration data, e.g., census data or survey data.}, but the aspects related to phone-to-human density transformation problem fall outside the scope of the present contribution.

\subsection{Data processing workflow}
In this section, we provide a general modular  view of the data processing workflow. For the sake of generality,  we make an abstraction effort and  distinguish, for each module, the  \emph{logical task} to be performed (\emph{what} the module does) from the particular method that may be adopted to perform such task  (\emph{how} the module does it). In this way, the proposed general workflow can be particularized to represent a broad set of different specific methodologies, allowing to explore  various trade-offs in the methodological solution space for the problem at hand.

\begin{figure*}[tb!]
\centering
\includegraphics[width=0.99\linewidth]{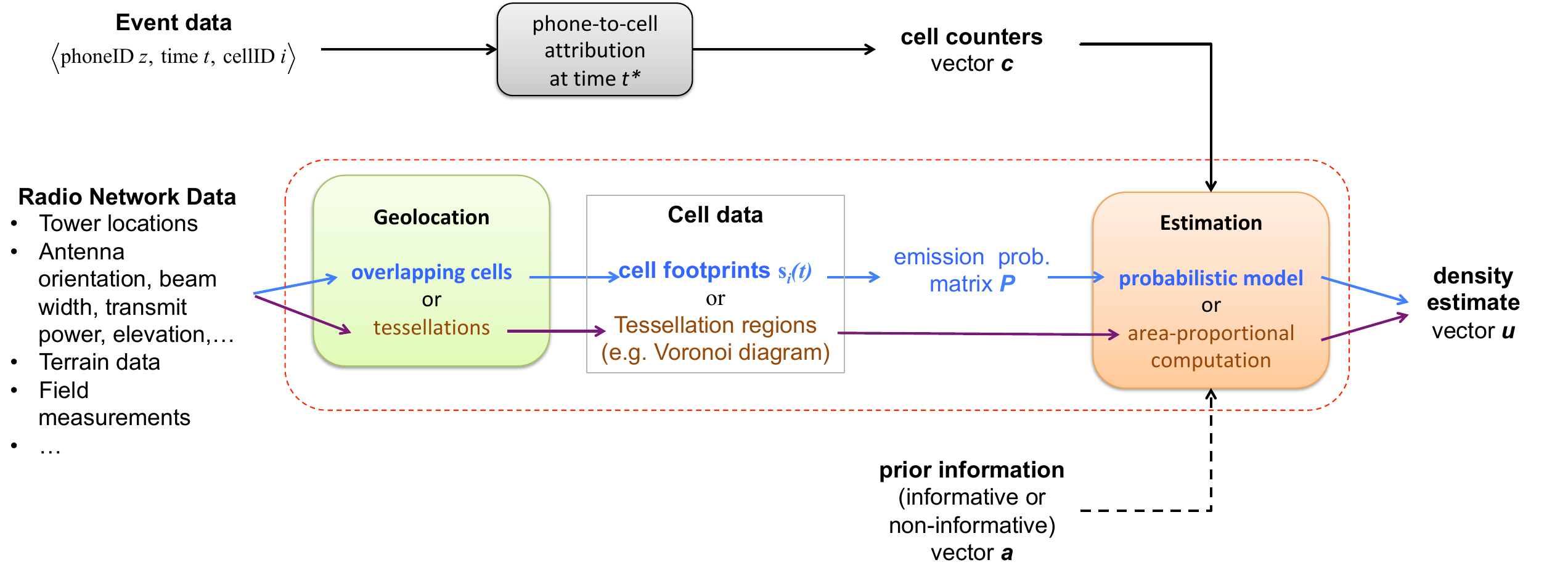}
\caption{Data processing workflow. The main focus of this work relates to the estimation module, based on a probabilistic model, associated to geolocation methods allowing  overlapping cells (upper path in the central box). The workflow for  tessellation methods (bottom path) may be considered as a particular, degenerate case of the previous method, where the emission probability matrix $\bm P$ reduces to the identity matrix and therefore can be omitted. Due to the absence of the stochastic component, the estimation module degenerates into a simple area-proportional computation. }
\label{fig:workflow}
\end{figure*}

\begin{figure}[tb!]
\centering
\includegraphics[width=0.99\linewidth]{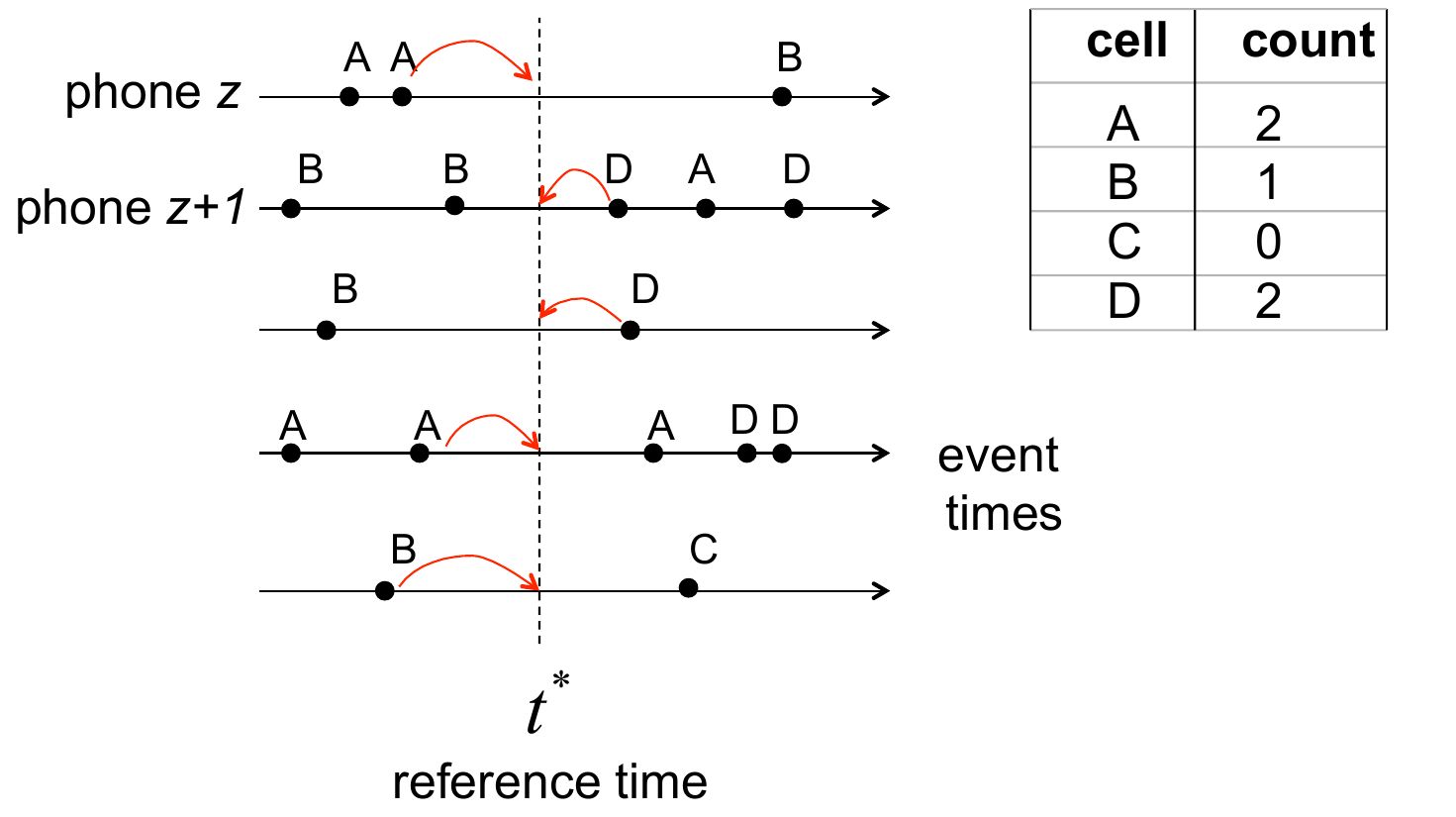}
\caption{Illustration of phone-to-cell association based on closest event.}
\label{fig:merging}
\end{figure}

The general workflow is sketched in Fig. \ref{fig:workflow}, separately for the two families of geolocation methods that are introduced below.

The processing workflow  takes two logically distinct sources of data in input, namely ``event data'' and  ``radio network data''. These data are pre-processed separately and then combined.\\

\par {\bf Event data} result from transaction events between the mobile phone and the network infrastructure, e.g., calls, data exchange or  signaling procedures.
    A mobile phone $z$ that at time $t$ interacts with the mobile network through the radio cell $i$ generates an event record $\langle z,t,i \rangle$.
    In practice $z$ represents the phone pseudonym\footnote{The phone pseudonym is typically based on some non-reversible hash of the International Mobile Subscriber Identifier (IMSI).}, $t$ the event timestamp and $i$ the radio cell identifier\footnote{
       Under  certain  conditions, e.g.,  signaling data captured on the Radio Access Network (RAN) links, additional variables which could help identify the mobile phone location at sub-cell level may be available in addition to the cell identifier $i$, e.g., the so-called Timing Advance (TA)  allowing to determine upper and lower bounds on the distance between the phone and the base station antenna, or radio measurements from neighboring cells as used, e.g., in Location Based Systems (LBS).  These additional variable(s)  should be considered absorbed by the the data element $i$ that, in such cases, should be reinterpreted as the union of the cell identifier and any other available variable carrying spatial information, e.g., TA or LBS. The use of additional information from the radio interface is still relatively infrequent (see \cite{calabrese} for a prominent example) but could become more important in future 5G deployments. Here, we do not elaborate further this scenario, and maintain the variable $i$ to denote solely the radio cell identifier.}.
    Depending on how the data collection system is configured, only a subset of all events are captured into the data set. Call Detail Records (CDR) and the more informative signaling data fall in this category. Recall that we aim at inferring spatial density with reference to particular time instant $t^*$ from the event records collected around that time. The phone-to-cell attribution module reduces the set of event data to a {\bf vector  of cell counters}  $\bm c \eqdef [c_1\ c_2\ \cdots\ c_I]^\mathsf{T}$ wherein the  $i$th element $c_i$ denotes the number of phones counted in cell  $i=1,\ldots, I$. A simple method to perform the phone-to-cell attribution is sketched in Fig. \ref{fig:merging}: for each generic phone $z$, the event record that is temporarily closest to the reference time $t^*$ is selected and the phone location is attributed to cell index $i$ therein. In this way, increasing the frequency of observation --- as done by considering full signalling data instead of less frequent CDR --- reduces the \emph{measurement noise} associated to the risk that the (unknown) phone position at the reference time $t^*$ falls outside the attributed cell $i$.\\
    
 \par {\bf Radio network data} relate to the deployment and configuration of the radio network infrastructure. Such data may be more or less detailed, ranging from merely a set of antenna tower locations to a very detailed set of full radio configuration parameters, e.g., antenna orientation, transmission power, beam width, frequency, etc. Based on the available radio network data, the 
``geolocation'' module derives a set of {\bf cell data} that aim at binding the radio cells to the spatial territory. 
 
 \subsection{Geolocation approaches}
 The availability of more complete radio network data offers the possibility to adopt more sophisticated geolocation methods in order to produce more detailed cell data. Conversely, if only minimal  radio network data are available, there is no choice but to adopt very simple geolocation methods to produce undetailed cell location data. For the sake of generality we need to identify a cell data format that can accommodate the whole range of possibilities, including but not limited to very simple ones. To  this aim we introduce the notion of ``cell footprint" as explained below. 

Let us consider a discretization of the geographical territory in a regular grid where  grid units are indexed in $j =1,\ldots,J$. For a generic cell $i$ at the reference time\footnote{The temporal component of cell data cannot be ignored, since  network coverage is not static but changes in time due to the continuous addition and deletion of radio cells, to the manual or automatic (re)optimization of radio parameters and other phenomena that are intrinsic to the functioning of mobile networks such as, e.g., the so-called ``cell breathing" effect in 3G, whereby the cell coverage area shrinks as the local network load increases  \cite{Mishrabook}.} $t^*$  let the non-negative quantity $s_{ij} \geq 0$ encode the degree by which  cell $i$ is expected to ``cover" the point $j$ in space.  The collection of such variables for a single cell $i$, formally   $\bm s_i \eqdef [s_{i1} \cdots s_{iJ}]^\mathsf{T}$, represents the ``cell footprint". 
  Depending on the assumed geolocation model, the  variables $s_{ij}$'s can be restricted to take binary values, thus leading to on/off flat coverage patterns for each individual cell,   or alternatively they could be allowed to take continuous values that capture the strength of the cell $i$ signal around point $j$ --- an approach first proposed in \cite{moblocworkingpaper,moblocpresentation} and then adopted also in \cite{ricciatoworkingpaper,salgado21}. 

If cell data are expressed in the form of cell footrpint, then the task of the geolocation module is to determine the set of cell footprints from the available radio network data. Radio cells that have identical footprints can be grouped together, with index $i$ now referring to the whole cell group.   
There are numerous geolocation methods, each making use of more or less information and with different degrees of sophistication. Every geolocation method  has, explicitly or implicitly, an underlying assumption about the process by which a generic mobile phone selects a radio cell to connect.
Generally speaking, following \cite{ricciatoworkingpaper} we can classify all possible approaches into two large families: 
\begin{itemize}
    \item {\bf Tessellations}: methods that partition the territory into a set of disjoint (non-overlapping) footprints associated to different cells or groups thereof.
    \item {\bf Overlapping cells}: methods that allow cell footprints to overlap.
\end{itemize}
The family of tessellation approaches may be seen as a particular (degenerate) case of the more general overlapping cells approach. 
With tessellations, cell data reduce to polygons or grid coordinates, and the general workflow  simplifies as sketched in Fig. \ref{fig:workflow}. 

\subsection{Geolocation methods based on tessellations}

The majority of earlier studies based on MNO data have adopted tessellation methods relying on the  Voronoi partitioning principle. We recall that, for a discrete set of $k$ points (Voronoi seeds) in a bounded planar region, a Voronoi partition divides the  region into a set of $k$ non-overlapping sub-regions (Voronoi polygons) whereby each point in space is associated to the closest seed. A Voronoi partition is well defined, in the sense that a single partition (Voronoi diagram) exists for a given set of seeds in the bounded region. However, there are multiple ways of projecting a set of mobile radio cells into a set of Voronoi seeds and multiple ways of bounding a region of interest, and each method results in a different Voronoi diagram for the same input set of cells. 

The simplest and by far most popular method takes antenna tower positions as seeds. As, in real networks, multiple radio cells have their antennas co-located on a single antenna tower, the number of unique seeds, hence polygons, is lower than the number of radio cells. This method does not require any  information about the radio cell configuration other than the antenna tower position, and is therefore very simple to implement. The implicit modeling assumption underlying this method is that mobile phones always connect deterministically to the closest antenna tower. This assumption ignores several fundamental aspects of mobile communications such as antenna directionality, power control, load balancing and  multi-layer radio deployments \cite{Mishrabook}.

Another variant, first considered in \cite{csis115,jrcstudy,ricciato16} and more recently also in \cite{bachir2019}, places seeds at the barycenter of radio cell coverage area. 
This  requires additional knowledge about the radio cell configuration parameters in case of directional cells (e.g., azimuth orientation, beam width, coverage range) in order to determine, at least roughly, the cell coverage area and then the barycenter.

The method presented by the Belgian operator Proximus in \cite{meersman2016} is more articulate. First, they distinguish between large cells (macro and micro cells)  and small cells (femto and pico cells) and apply Voronoi partitioning only to the large cells. 
Second,  in order to take into account the directionality of cell sectors, they place $N$ Voronoi seeds in the vicinity of each $N$-sector antenna (typically $N=3$ for $120^\circ$ sectors) with a small offset in their respective azimuth directions.

Generally speaking, all variants of the Voronoi approach have the disadvantage that even a single cell change, e.g.,  addition, removal or shift of its associated seed, produces a change in the neighboring cell footprints and consequently requires a new computation of the entire Voronoi diagram. 
Departing from the Voronoi approach, other forms of tessellations may be obtained by mapping each cell to one particular unit of a predefined partition, e.g., a regular square grid  (as in \cite{fiore2016}) or a variable resolution quadtree \cite{quadtree2017}. 

Whether based on Voronoi or an independent fixed grid, the more sophisticated variants of tessellation methods require additional information about the radio cell configuration, beyond the antenna location. However, as more detailed cell data are made available for the geolocation process, the limitation of considering non-overlapping footprints seems to be less justified, motivating the interest for methods accounting for the intrinsic overlapping nature of real-world radio cells. 

\subsection{Geolocation methods based on overlapping cell locations}

This family of methods explicitly takes into account the fact that real-world radio cells overlap by design.
To the best of our knowledge, the first work to  consider an overlapping cells approach for the problem of density estimation was proposed by Ricciato et al. in \cite{jrcstudy} and \cite{ricciato16} (see also the earlier work \cite{ricciato15} for a different application). Therein, the authors consider cell footprints that are ``flat", i.e., for a given radio cell $i$, a point $j$ in space is either included or excluded from its coverage area.

A more sophisticated approach was developed recently by Tennekes et al. in \cite{moblocworkingpaper,moblocarxiv} (see also the earlier presentation \cite{moblocpresentation}) and implemented in the \texttt{mobloc} R package \cite{mobloc}. 
This approach considers a non-uniform profile whereby the parameters $s_{ij}$'s are continuous and vary within the cell footprint, having a physical interpretation connected to  the received signal strength. 
The propagation model in \cite{moblocworkingpaper} also uses elevation maps and land use --- the latter only to estimate the path-loss exponent.

The determination of the cell footprints $\bm{s}_i$'s is a critical modeling task: the set of possible implementation solutions is wide,  ranging from elementary  geometric models as in \cite{jrcstudy,ricciato16} to more articulate but still parsimonious parametric propagation models as done in \cite{moblocworkingpaper}, and even  more sophisticated (and less parsimonious) empirical  numerical models that take into account more detailed data about the territory (elevation maps, type of buildings) and about the radio network dynamics (power control, inter-cell interference, etc.), possibly reusing data and tools that are already available and used regularly for radio network planning and optimization tasks.   

Regardless of how the cell footprints $\bm{s}_{ i}$'s are empirically determined, the next step in the modeling process is to encode such information into a probabilistic data generation model that will then serve as basis for the development of an estimation (inference) method. In the next section, we first introduce the probabilistic model, and then explain how this model is linked to the geolocation method of choice in the specific context of MNO data analysis.

\subsection{Probabilistic model}
\label{sec:modelstatement}

For the rest of this paper, we assume that the  geographic territory is discretized into a regular square grid. This assumption enables the use of vector notation and thus simplifies the formalism  as well as the software implementation of the considered methods.

We focus on the estimation problem at a given reference time,  therefore the temporal dimension can be dropped from the formalism. 
We resort to the term ``tile" to refer to the individual square grid unit, reserving the term ``cell"  to denote the radio cell constituting the mobile (or \emph{cellular}) radio network. 

Let the $j$th element $u_j$ of the column vector $\bm u \eqdef [u_1  \cdots u_J]^\mathsf{T}$  denote the unknown number of mobile phones in tile $j=1,\ldots, J$.
Let the $i$th element $c_i$ of the
column vector $\bm c \eqdef [c_1  \cdots c_I]^\mathsf{T}$ denote the observed number of mobile phones  counted in cell  $i=1,\ldots, I$ (cell count vector). We denote the total number of phones across all cells by 
$\myctot \eqdef \sum_{i=1}^I{c_i}  = \bm{1}_I^\mathsf{T} \bm{c}$ where the symbol $\bm{1}_k$ denotes a  column vector of size $k$ with all elements equal to 1.
We denote by $p_{ij}\in [0,1]$ the probability that a generic phone placed in tile $j$ will be detected  (counted) in cell $i$. In other words, $p_{ij}$  represents the conditional probability:
\begin{equation}
p_{ij} \eqdef \text{Prob}\left\{ \text{detected in cell $i$} \; | \;   \text{placed in tile $j$} \right\}.
\label{eq:pij}
\end{equation}
The $p_{ij}$'s are called ``emission probabilities'' in  the field of emission tomography \cite{shepp77} and we retain here the same term. Their value can be instantiated based on the cell data in output from the geolocation model, as explained below in Subsection \ref{sec:emissionprob}. 

For the sake of a more compact notation we gather the individual probabilities $p_{ij}$'s into a matrix $\bm P_{[I \times J]}$.
The elements of $\bm P$ sum up to one along columns, formally ${\bm 1}_I^\mathsf{T} \bm P   = {\bm 1}^{\mathsf{T}}_J $, meaning that the matrix $\bm P$ is column-stochastic.

The (measured) cell count vector $\bm c$ can be interpreted as the single realization of a random vector $\tilde{\bm c}$ whose expected value is given by:
\begin{equation}
\overline{\bm c} \eqdef \mathrm{E}[ \tilde{\bm c}] = \bm P \bm u.
\label{eq:forward}
\end{equation}
In the estimation problem we must solve for estimand $\bm u$ given the vector of measurement data $\bm c$, representing the single available observation of $\tilde{\bm c}$, and the model matrix $\bm P$.   This being  a  type of inversion problem,
the estimate  $\hat{\bm u}$ 
can be written in general as:
\begin{equation}
\hat{\bm u} =  g( \bm P, \bm c)
\label{eq:backward}
\end{equation}
where $g(\cdot)$ denotes the estimator of choice. 
It is evident that the estimand variables must be constrained to be non-negative, i.e.,
$u_i \geq 0, \; \forall i$. However, we relax the constraint that such variables should be integers since the rounding error can be safely neglected vis-à-vis other sources of uncertainty.

In practical deployments, the number of tiles is (much) larger than the number of cells, i.e., $J \gg I$, and moreover $J\gg 1$ (for instance, $J=160,000$ and $I$ around 600 in the simulation scenario considered in Section \ref{sec:numres}).  Therefore, even if $ \overline{\bm{c}}$ were perfectly known, the direct inversion of Eq. \eqref{eq:forward} would constitute an under-determined problem. 
For that reason, the estimation problem in this case is affected by issues of \emph{structural non-identifiability}, as per the definition given in \cite{raue13}. 
Any additional external information that is available to help the estimation process (e.g., prior distributions, spatial constraints derived from geographical maps, known structural properties of the desired solution) can be embedded in the estimator $g(\cdot)$ to resolve, or at least reduce, the ambiguity among multiple solutions, as we will elaborate upon later.

If two generic tiles $j_1$ and $j_2$ yield equal emission probabilities for all cells, i.e., $p_{ij_1}=p_{ij_2} \; \forall i$, and are associated with the same prior values (in case prior information is used in the estimate) then they are indistinguishable from each other. In such a case, their respective estimand variables $u_{j_1}$ and  $u_{j_2}$ are perfectly \emph{collinear} in the estimation problem and there is no way to resolve differences between them. Therefore, it makes sense to merge both tiles into a single super-tile, and then compute a single estimate for the latter\footnote{The concept of super-tile that is introduced here  algebraically  corresponds to the concept of  
``section''  introduced geometrically in the earlier study \cite{ricciato16}.}. Algebraically, this corresponds to merging identical columns of  matrix $\bm P$. We refer to this operation by the term ``consolidation''. Note that the consolidation process does not imply that the resulting (consolidated) matrix is necessarily full rank, i.e., it does not guarantee that the resulting consolidated instance of the problem is fully identifiable.

\subsection{Linking cell data to emission probabilities}
\label{sec:emissionprob}
The emission matrix $\bm P$ plays the role of a known input  parameter to the estimation problem. Its value must be determined from the cell data in output from the geolocation module (ref. to Fig. \ref{fig:workflow}). To this aim, a natural choice is to follow the approach first proposed by \cite{moblocworkingpaper} and later adopted by other work \cite{salgado21,essnetWPIdeliverableI3}. 
This model assumes that a mobile phone in tile $j$ selects cell $i$ with a probability that is proportional to its signal dominance \emph{relative}  to the other concurrent cells in the same tile, formally:
\begin{equation}
    p_{ij} = \frac{s_{ij}}{\sum_{m=1}^{k_j}{s_{mj}}}
    \label{eq:s2p}
\end{equation}
with $k_j$ denoting the number of concurrent cells at tile $j$.

This general model includes, as a particular case, the approach  considered in the earlier work \cite{ricciato16} where cell coverage was assumed to be on/off. In this case the  variables $s_{ij} \in \{0,1\}$ are limited to take binary values, hence Eq. \eqref{eq:s2p}  leads  the non-zero elements to take fractional values $p_{ij} = 1/k_j$, meaning that all $k_j$ cells covering tile $j$ have the same probability of being selected.

Tessellations may be seen as a further special case of the particular on/off coverage case above:  as each point in space is covered by exactly one and only one cell ($k_j=1$) the emission probabilities reduce to  binary values   
$p_{ij} \in \left\{0,1 \right\} $.  In other words, when cells are mutually non-overlapping, 
the data generating process becomes \emph{deterministic}: the binary model matrix $\bm P$ after consolidation reduces to the identity matrix, and the estimation problem becomes trivial
as far as no external information is taken into account (such as, e.g., prior information in Bayesian settings as considered in \cite{benjamin2018}).
In this contribution, we focus  on the more general (and non trivial) case of overlapping cells, with the understanding that the proposed solution will be applicable also to  the special case of tessellations.

\section{Estimation methods from literature}\label{sec:soa}
Whatever approach is chosen to instantiate the model matrix $\bm P$ (input parameter), a suitable resolution procedure is needed to compute the function $g(\cdot)$ in Eq. \eqref{eq:backward}.
This is the focus of the remaining part of this contribution. 
In this section we present three different solutions proposed in recent literature. To the best of our knowledge, no other solution beyond these three was previously considered  for the problem at hand. 

\subsection{MLE-Multinomial}\label{sec:mlemulti}
The method elaborated in \cite[Section 5.3]{ricciato16} (previously appearing in \cite{jrcstudy}) derives a Maximum Likelihood Estimator (MLE) based on the hierarchical generative model sketched in Fig. \ref{fig:hierarchicalmultinomial}.
\begin{figure}
\centering
\subfigure[Multinomial]{\includegraphics[width=0.42\linewidth]{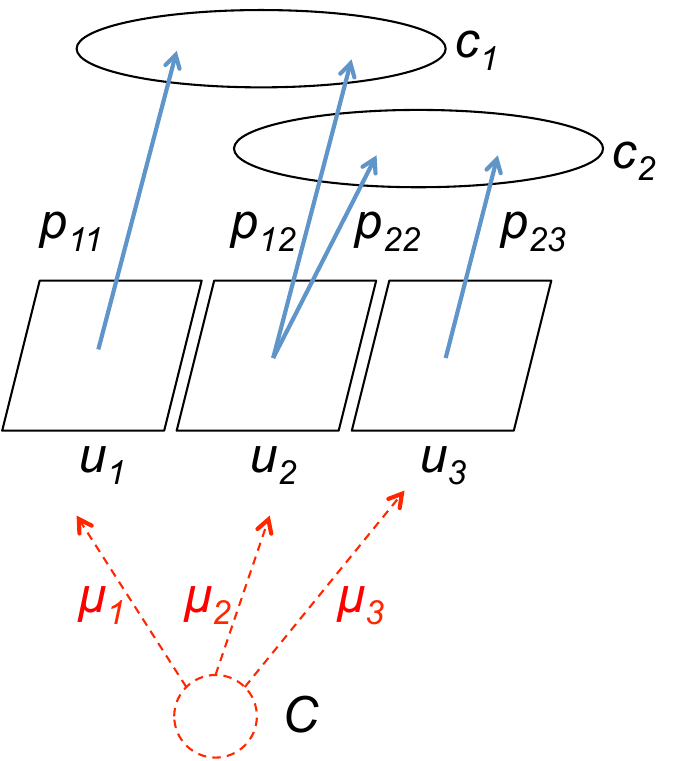}\label{fig:hierarchicalmultinomial}}
\subfigure[Poisson]{\includegraphics[width=0.55\linewidth]{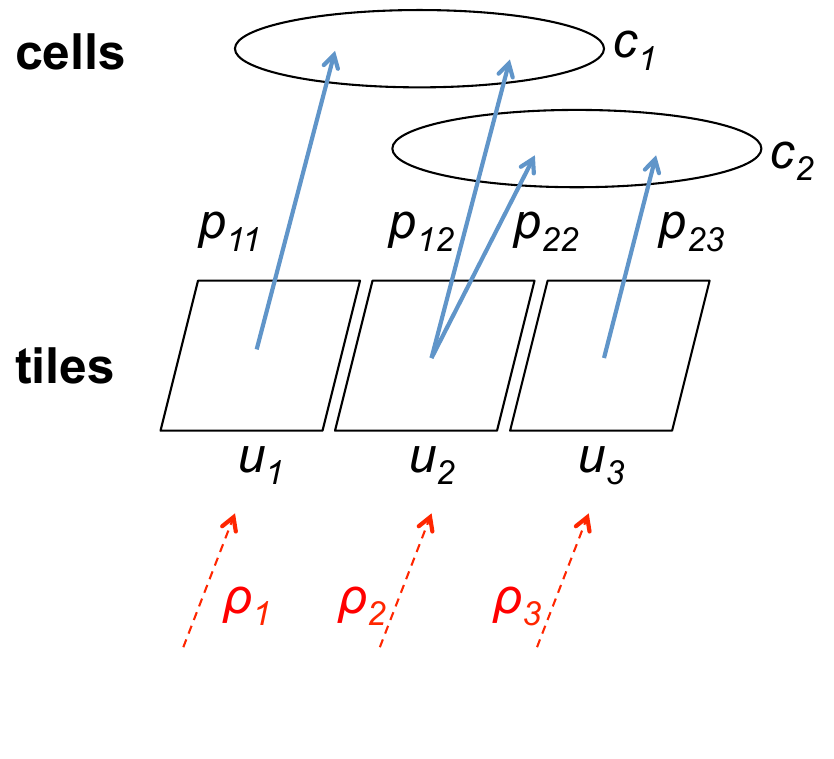}\label{fig:hierarchicalpoisson}}
\caption{Hierarchical generative models.}
\label{fig:hierarchical}
\end{figure}
In the first layer, starting from a single root pool, the $\myctot$ phones are randomly and independently allocated to the $J$ tiles
according to a vector of probabilities $\bm \mu \eqdef [\mu_1 \cdots \mu_J]^\mathsf{T}$.  Clearly, $\bm \mu$  is a vector with non-negative elements  summing up to unity, i.e.,  $\bm \mu \geq \bm 0$ and 
$ \bm{1}_J^\mathsf{T} \bm{\mu} =1$.
The resulting number of units in the tiles 
is  distributed as a multivariate
random vector $\tilde{\bm u}$ with Multinomial distribution (denoted by $\mathcal{M}$) of parameters $\myctot$ and $\bm \mu$, i.e., 
\begin{equation}
\tilde{\bm u} \sim  
\mathcal{M}\left(\myctot, \bm \mu \right)
\eqdef \frac{\myctot!}{\tilde{u}_1! \tilde{u}_2! \cdots \tilde{u}_J!} \prod_{j=1}^{J}{ \left( \mu_j\right)^{\tilde{u}_j}}
\label{eq:multinomial1}
\end{equation}
with  $\tilde{\bm u} \geq \bm 0$ and $\bm{1}_J^\mathsf{T} \bm{\tilde{u}}=\myctot$.
Therefore by construction the mean value can be written as $\overline{\bm u} \eqdef \E[\tilde{\bm u} \,|\, \myctot, \bm \mu]=\myctot \bm \mu$.
In the second layer (ref. Fig. \ref{fig:hierarchicalmultinomial}) the units are assigned randomly and independently from tiles to cells according to the emission matrix $\bm P$.
Due to the independence of the random assignments at the two layers --- from the root pool to tiles, and from tiles to cells --- the  random vector $\tilde{\bm c}$ also has a Multinomial distribution with parameters  $\myctot$ and $\bm P \bm \mu$:
\begin{equation}
\tilde{\bm c} \sim  
\mathcal{M}\left(\myctot, \bm P \bm \mu \right)
\eqdef \frac{\myctot!}{\tilde{c}_1! \tilde{c}_2! \cdots \tilde{c}_I!} \prod_{i=1}^{I}{ \left( \sum_{j=1}^J{ \mu_j\, p_{ij}}\right)^{\tilde{c}_i}}.
\label{eq:multinomial2}
\end{equation}
Thus, the log-likelihood function for an observed value  $\bm c$ is derived as (omitting irrelevant terms):
\begin{equation}
 \mathcal{L}(\bm \mu ; \bm c, \bm P) =  \sum_{i=1}^I{ c_i \log{ \sum_{j=1}^J{\mu_j p_{ij}}}}
 = \bm c^\mathsf{T} \, \log{ \bm P  \bm \mu}
 \label{eq:loglike0}
\end{equation}
wherein we have used the  logarithm of the vector notation $\log{\bm x} \eqdef \left[ \log{x_1} \ \log{x_2} \ \cdots \right]^\mathsf{T}$ to refer compactly to the vector of element-wise logarithms.  
Let $\hat{\bm \mu}_\text{MLE}$ denote the value of the probability vector $\bm \mu$ that maximizes the log-likelihood Eq. \eqref{eq:loglike0}  subject to the constraints  $\bm \mu \geq \bm 0 $ and $\bm{1}_J^\mathsf{T} \bm{\mu}  = 1$, formally:
 \begin{equation}
\hat{\bm \mu}_\text{MLE}
= \argmax_{ \substack{ \bm{1}_J^\mathsf{T} \bm{\mu} = 1 \\  \bm \mu \geq \bm 0} }{ \, \bm c^\mathsf{T}  \log{ \bm P  \bm \mu}  }.  
\label{eq:MLfwithlambda}
 \end{equation}
 By definition $\hat{\bm \mu}_\text{MLE}$ represents the Maximum Likelihood (ML) estimate for $\bm \mu$ given the observed data $\bm c$, but still we have to provide a point estimate for $\bm u$. A natural choice (also taken by Shepp and Vardi in their seminal paper \cite{shepp77} for the method presented in the next subsection) is to take the corresponding  \emph{expected mean value} 
 as the final estimate,  formally:
 \begin{equation}
 \hat{\bm u}_\text{MLE} \eqdef \E[\tilde{\bm u} \,|\, \myctot, \bm \hat{\bm \mu}_\text{ML}] = \myctot \, \bm \hat{\bm \mu}_\text{MLE}
 \label{eq:MLML}
 \end{equation}
where we retain the label ``MLE" for simplicity.
This equation represents a simple rescaling of $ \hat{\bm \mu}_\text{MLE}$ by the factor $\myctot$. Therefore, with a simple variable substitution, the minimization Eq. \eqref{eq:MLfwithlambda} can be rewritten to deliver the final estimate $\hat{\bm u}_\text{MLE}$ directly,  leading to the following constrained optimization problem: 
 \begin{equation}
\hat{\bm u}_\text{MLE}  
= \argmax_{ \substack{ \bm{1}_J^\mathsf{T} \bm{u}= \myctot \\  \bm u \geq \bm 0} }{ \, \bm c^\mathsf{T}  \log{ \bm P  \bm u}  }  
\label{eq:MLfinal_u}
 \end{equation}
In the original papers \cite{jrcstudy,ricciato16} the minimization in Eq. \eqref{eq:MLfinal_u} is conducted via standard numerical solvers. No remark is made therein about the uniqueness (or lack thereof) of the solution. 

\subsection{MLE-Poisson}\label{sec:mlepois}

Recently, the authors of \cite{cbsnetmob2019} have considered to apply the ML estimator developed in the field of emission tomography by Shepp and Vardi in \cite{shepp77}  to this problem. 
Like the previous approach, this method is also based on a hierarchical generative model with two layers, as sketched in Fig. \ref{fig:hierarchicalpoisson}, but now the elements of $\tilde{\bm u}$ are modeled as independent Poisson (instead of Multinomial) random variables with parameters $\bm \rho \eqdef [\rho_1 \cdots  \rho_J]^\mathsf{T}$, i.e., $\tilde{u}_j \sim \mathcal{P}(\rho_j)$ ($\mathcal{P}$ denoting the Poisson distribution). 
Following the same reasoning that led to Eq. \eqref{eq:MLML}, the ML estimate of $\bm \rho$ is taken as the estimate for $\bm u$ (see \cite[p. 113]{shepp77}). 
The log-likelihood function is not given explicitly in  \cite{shepp77} but it is proved that the ML estimate can be computed iteratively through an Expectation Maximization (EM) procedure: at the generic iteration $m$ the new estimate $\hat{u}_j^{m+1}$  is computed from the previous estimate $\hat{u}_j^{m} $ according  to the following formula (see \cite[Eq. (2.13)]{shepp77} or, equivalently, \cite[Eq. (2)]{cbsnetmob2019}):  
\begin{equation}
\hat{u}_j^{m+1} = \hat{u}_j^{m}  \cdot \sum_{i=1}^I{  c_i \frac{p_{ij} }{ \sum_{k=1}^J{ p_{ik} \hat{u}_k^{m} }} }.
\label{eq:shepp}
\end{equation}
The authors of the original paper \cite{shepp77} warn that the initialization point should not contain zero elements, and by default assume a flat (uniform) initial solution $\hat{u}_j^0=\frac{\myctot}{J}$, $j=1,\ldots,J$.

\subsection{A simple estimator based on Bayes' rule (SB)}\label{sec:bayest}

The simple procedure presented hereafter was adopted in the \texttt{mobloc} R package developed by Tennekes et al. \cite{mobloc} and elaborated in \cite{moblocworkingpaper} (see also \cite{moblocpresentation}). 
We shall refer to this method as the ``Simple  Bayes-rule estimator" (SB for short). 
Let $q_{ji}$ denote the  conditional probability:
\begin{equation}
q_{ji} \eqdef \text{Prob}\left\{ \text{placed in tile $j$}   \; | \;   \text{detected in cell $i$} \right\}.
\label{eq:qji}
\end{equation}
Note the inversion of the conditioning direction between $q_{ji}$ and  $p_{ij}$ defined earlier in Eq. \eqref{eq:pij}. 
Let 
\begin{equation}
\alpha_{j} \eqdef \text{Prob}\left\{ \text{placed in tile $j$}  \right\}
\end{equation}
denote the  prior probability that a single generic unit falls in tile $j$, before observing the measurement data. 
Recalling the  Bayes rule 
\begin{equation}
\text{Prob}\left\{ j  \, | \,  i  \right\}   = \frac{  \text{Prob}\left\{ i  \, | \,  j  \right\}  \cdot \text{Prob}\left\{   j  \right\} }{ \text{Prob}\left\{   i  \right\} }
\label{eq:bayes}
\end{equation}
it follows that 
\begin{equation}
q_{ji} 
= \frac{ p_{ij}  \, \alpha_j }{\sum_{k=1}^J{ p_{ik}  \, \alpha_k }}.
\end{equation}
Therefore,  the 
estimate in each tile $j$ is computed directly as
\begin{equation}
    \hat{u}_{j,\text{SB}} = \sum_{i=1}^I{q_{ji} \, c_i }= \alpha_j \, \sum_{i=1}^I{  c_i  \frac{ p_{ij}  }{\sum_{k=1}^J{ p_{ik}  \, \alpha_k }} }.
        \label{eq:SBelement}
\end{equation}

\section{Insights from existing estimators}\label{sec:newresults}
In this section we reinterpret  previous solutions and present our new results. More specifically, the analytical insight presented in this Section will serve as the basis for the development of a novel estimation method in  Section \ref{sec:map}.

\subsection{The prior vector}
\label{sec:remarksinitial}

All three methods reviewed above provide a point estimate $\hat{\bm u}$ in output,  based on the model matrix $\bm P$ and observed data $\bm c$ in input. To do so, they all require, implicitly or explicitly, the provision  of an ``initial point'' as input to the computation. The initial point takes either the form of a (stochastic) vector of $J$ non-negative elements summing up to unity, hereafter denoted by  $\bm \alpha \eqdef [\alpha_1 \cdots  \alpha_J]^\mathsf{T}$ (with $\bm{1}^\mathsf{T}_J\bm{\alpha} =1$ and $\bm \alpha \geq \bm 0$) or, equivalently, of its rescaled version 
$\bm a \eqdef \myctot \bm \alpha$ (with $\bm{1}^\mathsf{T}_J \bm{a}=\myctot$ and $\bm a \geq \bm 0$). Such  an initial vector represents our ex-ante ``best guess"  about how the mobile units may be spatially distributed, \emph{before seeing the data}; in this sense, as already noted by Shepp and Vardi in \cite{shepp77}\footnote{Quoting from \cite{shepp77}: ``As a point of philosophical interest, the choice of the initial estimate is somewhat akin to the choice of a Bayes prior but there is actually no prior measure.''}, it represents a sort of prior information, and for this reason we shall refer to $\bm \alpha$ (or equivalently to $\bm a$) as the ``prior vector''.
Note however the difference between the notions of ``prior vector'' $\bm{\alpha}$ (or equivalently $\bm{a}$) and ``prior probability distribution'' $\mathbb{P}(\bm{u}; \bm{\alpha})$: the former is a vector of scalars that can serve as parameters for the latter, which is conversely a function of the data $\bm{u}$ parameterized in $\bm{\alpha}$ and having the properties of a probability distribution (i.e., taking on positive values and summing up to one). This distinction will become more evident in Section \ref{sec:map}, where the Bayesian MAP and related estimators are introduced.

The role of the (non rescaled) prior vector $\bm \alpha$  is explicit in the SB estimator Eq.  \eqref{eq:SBvector} where its elements represent prior probabilities (per tiles).
As for MLE-Poisson, the (rescaled) prior vector $\bm a$  serves as the initial point for the iterative procedure Eq. \eqref{eq:shepp} by setting $\hat{\bm u}^0 = \bm a$.
Similarly,  $\bm a$  is the natural starting point for the numerical minimization process in Eq.  \eqref{eq:MLfinal_u} to compute the  MLE-Multinomial estimate. In summary, all three methods rely on a prior vector (rescaled or not) as the initial point for the computation. 

A point of caution is needed in determining the value of the  prior vector, particularly  concerning its zero elements. It can be immediately recognized that setting the initial value $\alpha_j=0$, or equivalently $a_j=0$, for a generic tile $j$ will force that tile to maintain a zero value in the final solution. This holds true both for the EM iterative procedure based on Eq.  \eqref{eq:shepp}, since zeros are stable points of the iteration due to its  multiplicative structure, and for the direct computation of SB based on Eq.  \eqref{eq:SBvector}. In other words, with both these methods, namely SB and MLE-Poisson with EM, zeroing the $j$th element in the prior vector  overrides any possible contribution carried by the measurement data for the corresponding tile $j$. This is equivalent to completely excluding \emph{a priori} the corresponding tile $j$ from the computation. But if we intend to do so, it is certainly more practical to drop the variable associated with such tile  upfront, rather than instantiating a variable whose final value is fixed ex-ante to zero. By eliminating the zero tiles upfront, we can practically assume that all remaining elements in the prior vector take non-zero values. 

In practical settings, the prior vector can be instantiated based on land use data obtained by other sources, e.g., satellite imagery or land use surveys as done in \cite{benjamin2018,batista2020} in the context of Voronoi tessellations. Such data can give indications about which areas are more or less likely to host humans, and therefore may increase the accuracy of the final estimate. However, translating external data into a meaningful prior information still requires some modeling choices. For instance, whether an outdoor  leisure area is more or less likely to attract visitors compared to a densely built area depends upon contextual factors, e.g.,   weather conditions or calendar day (working day vs. holiday). For this reason, the use of such information should be dealt with judiciously in order to avoid introducing biasing errors.  

\subsection{MLE-Multinomial and MLE-Poisson}\label{sec:equivalenceMLE}

We show that the two MLE procedures  presented earlier in Section \ref{sec:mlemulti}
and \ref{sec:mlepois} are equivalent in the sense that they yield exactly the same set of solutions. Beforehand, we need a technical lemma summarized in the following Proposition.

\begin{proposition}\label{prop_convex}
The log-likelihood Eq.  \eqref{eq:loglike0} is concave, but not strictly concave  and the constraints are linear, so the overall problem is concave but not strictly concave. As a consequence, all stationary points are equivalent global maxima. 
\end{proposition}
\begin{proof}
See Appendix~\ref{App2}.
\end{proof}

In the following, the equivalence between MLE-Poisson
and MLE-Multinomial is established. The solution space of the latter can be obtained by Proposition \ref{prop_convex} as the set of stationary points of the Lagrangian function, i.e., the derivative of the objective function in Eq. \eqref{eq:MLfinal_u} augmented by the scaled version of the equality constraint. The first-order optimality condition is therefore
\begin{equation}
\frac{\partial}{\partial u_j}
\left(\bm c^\mathsf{T}  \log{ \bm P  \bm u} - \lambda (\bm{1}^\mathsf{T}_J\bm{u} -C ) \right)  =0 \label{eq:optimality}
\end{equation}
which can be shown (see Appendix~\ref{App3})  to yield, for each $j$, 
\begin{equation}
    \sum_{i=1}^I c_i \frac{p_{ij}}{\bm{p}_i^\mathsf{T} \bm{u}} = 1
\label{eq:supstable}
\end{equation}
where $\bm{p}_i^\mathsf{T}$ denotes the $i$th row of $\bm{P}$. Eq. \eqref{eq:supstable} is identical to the convergence (fixed-point) stability condition of Eq. \eqref{eq:shepp}, i.e., $\hat{\bm u}^{m+1}=\hat{\bm u}^{m}$. More formally, we have the following result.
\begin{proposition}\label{prop_optimality}
The optimality condition of the MLE-Multinomial in Eq. \eqref{eq:optimality}  is identical to the convergence condition of the EM procedure based on the MLE-Poisson, i.e., the two methods share  the same solution space.
\end{proposition}
\begin{proof}
See Appendix~\ref{App3}.
\end{proof}

Looking retrospectively, this equivalence could have been expected in the light of the close relationship between the Poisson and Multinomial distributions  \cite[Ch. V]{Taylorbook}: if $K$ independent random variables follow a Poisson distribution $x_k \sim \mathcal{P}(\rho_k)$, $k=1,\ldots,K$, then the distribution of the vector $\bm x \eqdef [x_1 \cdots x_K]^\mathsf{T}$ \emph{conditioned to the  sum} $\sum_{k=1}^K{x_k}=\bm{1}^\mathsf{T}_K \bm{x} = X$ is Multinomial $\mathcal{M}(X,\bm \mu)$ with $\mu_k=\frac{\rho_k}{\sum_{\iota=1}^K{\rho_\iota}}$.

As an additional result, we can easily recognize that the multiplicative factor appearing in Eq. \eqref{eq:shepp} equals the partial derivative of the objective function in  Eq. \eqref{eq:MLfinal_u}, i.e., 
\begin{equation}
    \frac{\partial   }{\partial u_j}  \bm c^\mathsf{T}  \log{ \bm P  \bm u}  = \sum_{i=1}^I{  c_i \frac{p_{ij} }{ \sum_{k=1}^J{ p_{ik} {u}_k }} }.
\label{eq:partialderiv}
\end{equation}
In other words, we can reinterpret the EM procedure in Eq. \eqref{eq:shepp} (derived for the Poisson generative model) as a purely multiplicative method to solve iteratively 
the optimization Eq. \eqref{eq:MLfinal_u} 
(derived for the Multinomial generative model)
starting from an initial guess\footnote{Such a purely multiplicative procedure in the form $u_j^{m+1} = u_j^{m} \cdot \psi_j^{m}$ is distinct from the standard gradient descent approach that, in general, involves an additive update in the form $u_j^{m+1} = u_j^{m} + \phi_j^{m}$.}.

\subsection{Maximum Likelihood Bounded Subspace (MLBS)}\label{sec:mlbs}

Recall that the vector $\bm c$ represents a \emph{single realization} of the random vector $\tilde{\bm c}$ with unknown mean $\overline{\bm c}$. Since no other measurements are available, in the absence of external information 
the  vector $\bm c$ represents a natural \emph{estimate of the  mean value $\overline{\bm c}$}. 
Replacing the unknown term $\overline{\bm c}$  with its  estimate $\bm c$ in the generative relation Eq. \eqref{eq:forward} leads to the constraint $\bm P \bm u = \bm c$. It is rather intuitive that, among all possible (non-negative) values of $\bm u$, those respecting this constraint, if they exist, are the ones that best conform with the available measurement $\bm c$ and with the model $\bm P$. This argument leads us to restrict the search for a ``good" estimate $\hat{\bm u}$ within the bounded subspace
\begin{equation}
   \mathcal{U} \eqdef \left\{ \bm u \; : \;  \bm P \bm u = \bm c , \, \bm u \geq \bm 0 \right\}. 
   \label{eq:U}
\end{equation}
Note that the constraint $\bm P  \bm u = \bm c$ absorbs the condition  $\bm{1}^\mathsf{T}_J\bm{u} =\myctot$ on the total count.
Since in our application $J\gg I$,  
matrix $\bm P$ is rank deficient  and the constraint $\bm P \bm u = \bm c$ admits multiple solutions in the variable $\bm u$, although it is not guaranteed that all the components of the latter are non-negative, i.e., MLBS may be empty. 
For non-empty MLBS, while we have somewhat restricted the search space, we still need to provide a criterion for selecting a single solution within that space unambiguously. This aspect is elaborated later in Section \ref{sec:ourestimator}. 

Though this way of reasoning is  heuristic, it turns out that the bounded subspace defined by Eq. \eqref{eq:U} is intimately connected to the MLE-Multinomial and MLE-Poisson procedures. 
In fact, the condition $\bm{P}\bm{u} = \bm{c}$ can be rewritten as  $\bm{p}_i^\mathsf{T} \bm{u} = c_i$, $\forall i$, implying that the optimality condition  Eq. 
\eqref{eq:supstable} is always verified. In particular, the following result holds true.
\begin{proposition}\label{prop3}
If non-empty, the bounded subspace $\mathcal{U}$ defined in Eq.  \eqref{eq:U} represents a set of  solutions for both MLE procedures, i.e., any point in $\mathcal{U}$ is also a solution of MLE-Multinomial and MLE-Poisson.
\end{proposition}
\begin{proof}
See Appendix~\ref{App4}.
\end{proof}
Following Proposition \ref{prop3}, we shall refer to $\mathcal{U}$ as the Maximum Likelihood Bounded Subspace (MLBS).

\subsection{Insights about the SB  estimator}\label{sec:sb}

\begin{proposition}\label{prop1}
The simple Bayes-rule estimator in Eq. \eqref{eq:SBelement} belongs to the linear-type subclass of Eq. \eqref{eq:backward}; in fact, it can be  rewritten in vector form as
\begin{equation}
    \hat{\bm{u}}_\text{SB}  =  \bm Q  \bm{c}, \quad 
    \bm Q_{[J \times I]} \eqdef \mathrm{diag}(\bm{\alpha}) \bm{P}^\mathsf{T} \mathrm{diag}^{-1}(\bm{P}\bm{\alpha}) 
        \label{eq:SBvector}
\end{equation}
where  $\mathrm{diag}(\bm{v})$ is a diagonal matrix containing the entries of $\bm{v}$  (and $\mathrm{diag}^{-1}(\bm{v})$ their reciprocals). 
An alternative form is
\begin{equation}
 \hat{\bm{u}}_\text{SB} =  \bm{\alpha} \odot \bm{P}^T (\bm{c} \oslash \bm{P} \bm{\alpha} ) \label{eq:SB_vec}
\end{equation}
where $\odot$ denotes the element-wise (Hadamard) product between two vectors, and likewise $\oslash$ denotes the element-wise division between two vectors.
\end{proposition}
\begin{proof}
Notice that $\sum_{k=1}^J p_{ik} \alpha_k$ and  $\sum_{i=1}^I \xi_i  p_{ij}$ are the $i$th and $j$th element of the vectors $\bm{P}\bm{\alpha} $ and $\bm{P}^\mathsf{T} \bm{\xi}$, respectively, where $\xi_i = c_i/\sum_{k=1}^J p_{ik} \alpha_k$. The thesis follows in a straightforward manner by rewriting in vector form $\bm{\xi} = \mathrm{diag}^{-1}(\bm{P}\bm{\alpha}) \bm{c}$ and exploiting the associative property of the matrix product.
\end{proof}

Proposition \ref{prop1} highlights that, owing to its linearity, the SB estimator is simple to compute and, compared to MLE, does not involve an iterative procedure.\\

It should be noted that SB may fall outside the MLBS (a numerical example is given below in Subsection \ref{propSBfirstiteration}) and therefore \emph{SB does not qualify in general as a ML solution}.

\begin{proposition}\label{propSBfirstiteration}
The SB solution coincides with the first point of the EM sequence after the first iterative step when the starting point is set to the prior vector $\bm \alpha$.  
In other words, the SB point lies  between the prior vector $\bm \alpha$ and the MLBS. 
\end{proposition}

\begin{proof}
It can be immediately verified by
comparing Eq. \eqref{eq:SBelement}  and Eq. \eqref{eq:shepp}, with   $\hat{\bm u}^m_j$ in  the latter replacing  $\bm \alpha_j$ in the former. 
\end{proof}

\begin{figure}[tb!]
\centering
\subfigure[Toy scenario]{\includegraphics[width=0.6\linewidth]{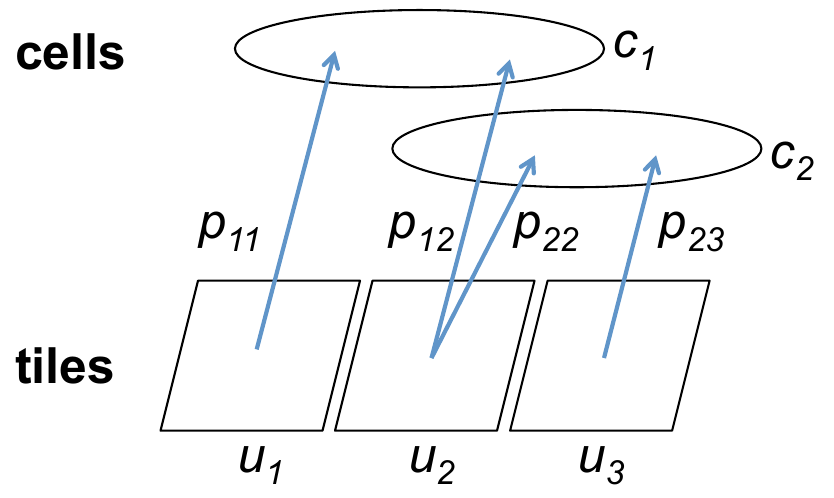}\label{fig:toyscenario}}
\subfigure[Log-likelihood surface]{\includegraphics[width=0.99\linewidth]{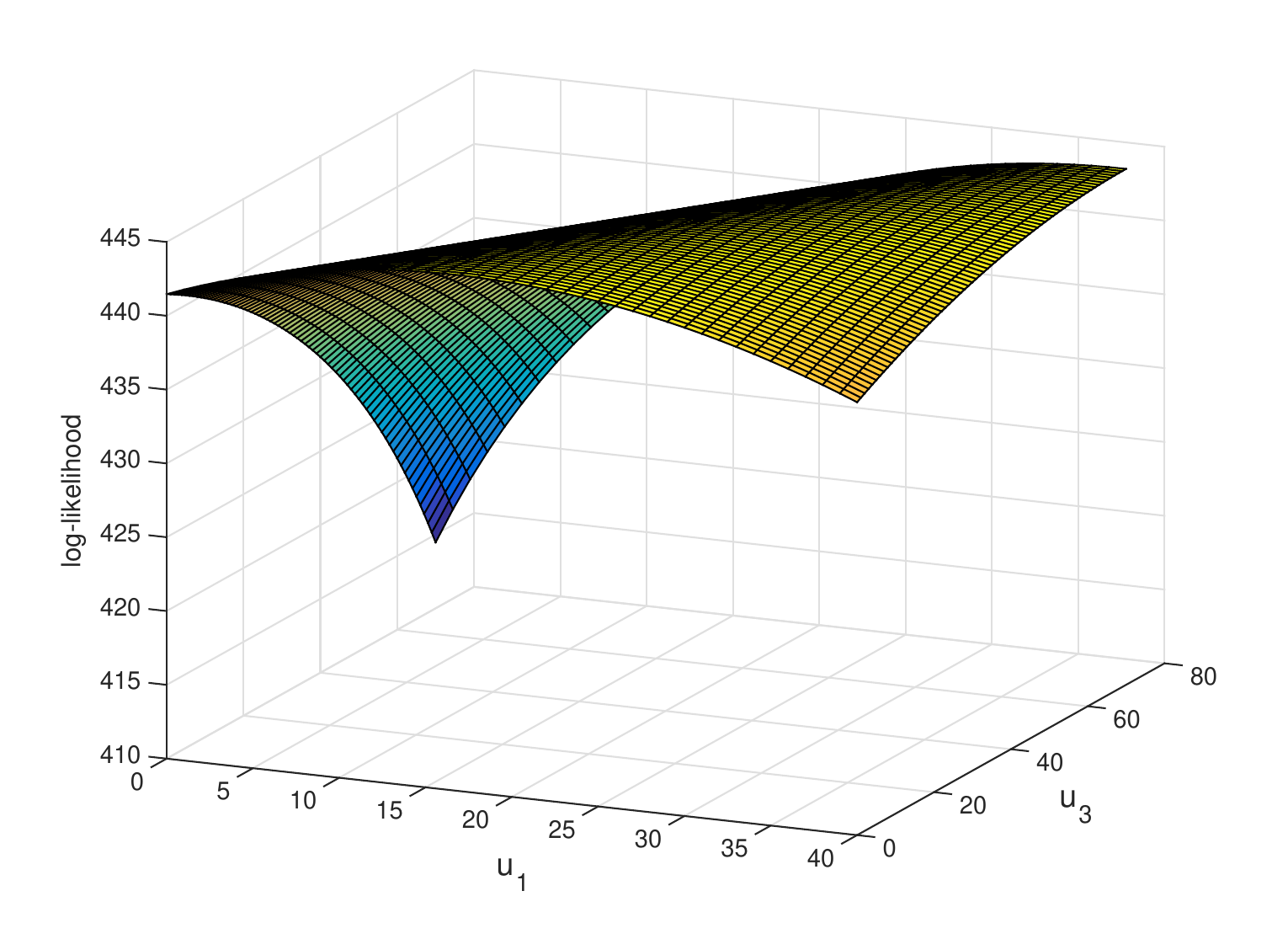}\label{fig:loglikesurface}}
\subfigure[Solution points]{\includegraphics[width=0.99\linewidth]{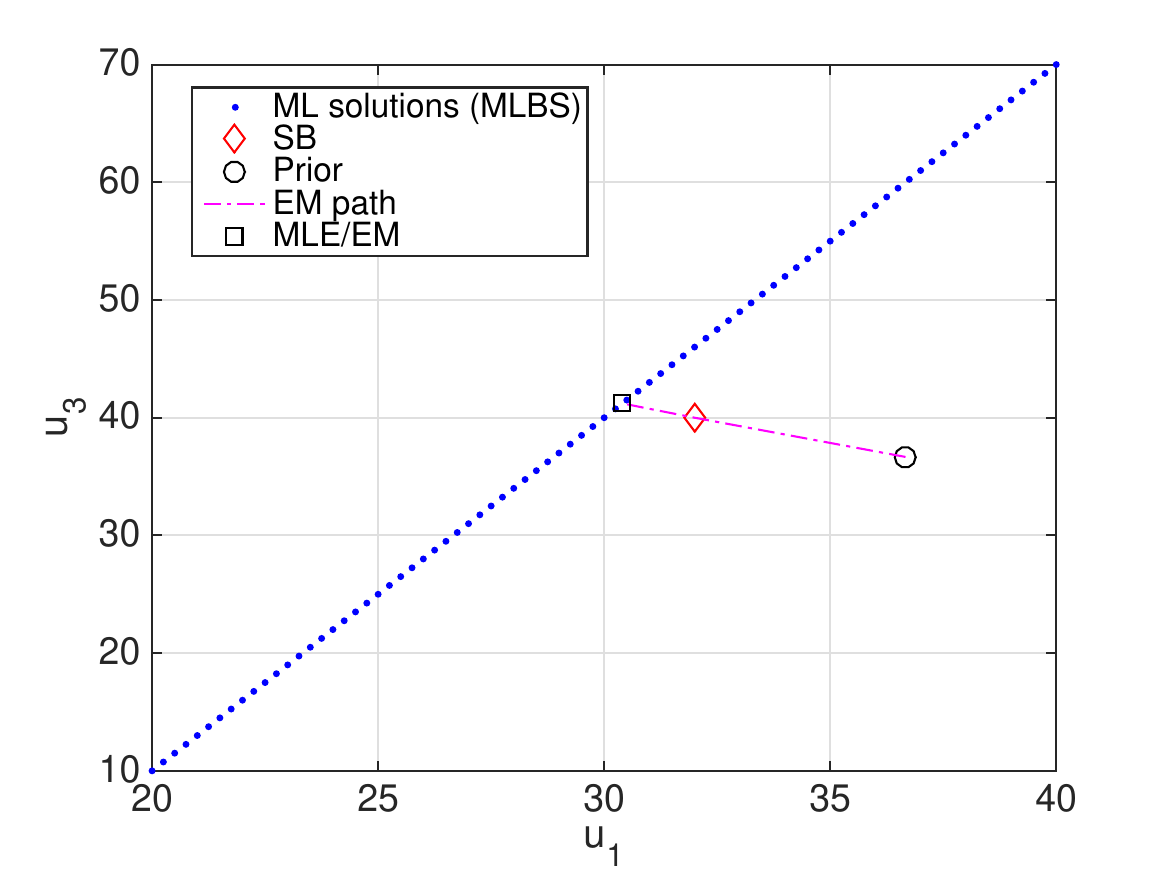}\label{fig:toyscenarioplanar}}
\caption{Log-likelihood surface and solution points for a toy scenario with $J=3$ tiles, $I=2$ cells, $C=110$ phones, count vector $\bm c = [40 \; 70]^\mathsf{T}$
and emission matrix $\bm P = 
\begin{bmatrix}
    1 & .25 & 0 \\
    0 & .75 & 1 \\
\end{bmatrix}
$.}
\label{fig:toyall}
\end{figure}

\subsection{Graphical illustration for a toy scenario}
\label{sec:toy}

To illustrate the above results graphically, we present  numerical results referring to a simple toy scenario where  $\myctot = 110$ mobile phones are split across $J=3$ tiles and connect to $I=2$ overlapping cells.
The observation count vector is $\bm c = [40 \; 70]^\mathsf{T}$
and the model matrix $\bm P = 
\begin{bmatrix}
    1 & .25 & 0 \\
    0 & .75 & 1 \\
\end{bmatrix}
$. In this toy instance there are only two degrees of freedom since the three variables $u_1$, $u_2$ and $u_3$ are constrained by the total sum  $\sum_{j=1}^3{u_j}=\myctot$, or equivalently $\bm{1}^\mathsf{T}_J\bm u  = \bm{1}^\mathsf{T}_I \bm c $. Therefore, we can map all admissible solutions $\bm u$ to an horizontal plane.
In the 3D plot of Fig. \ref{fig:loglikesurface} the horizontal dimensions represent $u_1$ and $u_3$, while the vertical dimension reports the value of the likelihood function in Eq.  \eqref{eq:loglike0}. The convexity of the  likelihood function (proven in Proposition \ref{prop_convex}) is graphically evident.
Notably, the loci of all global maxima of the maximum function lie on a segment, that is the 1D MLBS of the whole 2D solution space. 
In the planar plot of   Fig. \ref{fig:toyscenarioplanar}  the MLBS is marked by blue dots. The flat initial guess (for which $u_1=u_2=u_3=36.7$), the corresponding SB estimate  from Eq.  \eqref{eq:SBelement} and the EM solution from Eq.  \eqref{eq:shepp} are indicated by markers.

\section{Design of a Novel Estimator}\label{sec:map}

\subsection{Rationale}

In this section we derive a novel estimation approach that does not require an iterative procedure and can be solved in closed-form. The derivation starts from formulating the Maximum A Posteriori (MAP) estimator for the problem at hand, which is a typical approach to obtain point estimates  combining data and prior information. 
To the best of our knowledge, this is a novel contribution not presented earlier in previous literature.
We find however that the exact MAP estimator is impractical for the problem at hand, since the high computational complexity prevents its adoption in large problem instances. This motivates our effort to devise a novel alternative estimator. 

Generally speaking, the MAP estimator combines  information \emph{from the data}, captured by the likelihood term embedding  the measurement vector $\bm c$, with  information \emph{before the data}, captured by the prior probability distribution $\mathbb{P}(\bm{u}; \bm{\alpha})$ embedding the prior vector $\bm \alpha$ (or equivalently  $ \bm{a} = C \bm{\alpha}$), as discussed earlier in  Section \ref{sec:remarksinitial}. In fact, the posterior probability distribution
is generally composed of the product of the two components, namely  $\text{posterior} = \text{likelihood} \times \text{prior}$. The resulting solution is a trade-off, i.e., a compromise  between the two sources of information.
Notice that also the SB estimator introduced in Section \ref{sec:bayest}  aims at compromising between prior information and data, according to a Bayesian rationale. However, as discussed below, the MAP criterion makes a further step through the maximization of the posterior, which in fact leads to a more sophisticated estimator involving a computationally-intensive numerical optimization, while the SB is a simple (linear) estimator available in closed-form.

\subsection{The MAP estimator}\label{sec:MAP}

We derive the MAP estimator based on the hierarchical generative model with Multinomial distribution shown in Fig. \ref{fig:hierarchicalmultinomial}.  
The result is summarized in the following Proposition.

\begin{proposition}\label{prop_MAP}
The MAP estimator of $\bm{u}$ can be obtained as
\begin{equation}
\begin{split}
\hat{\bm u}_\text{MAP} & =  \argmax_{ \substack{ \bm{1}^\mathsf{T}_J \bm u  = \myctot \\  \bm u  \geq \bm 0} }{ \big\{ \bm c^\mathsf{T} \log{ \bm P  \bm u}  + \bm{u}^\mathsf{T}\log{\bm{\alpha}} -  \sum_{j=1}^J{\log{u_j!}}  \big\} } 
\end{split}
 \label{eq:MAPdevel}
\end{equation}
\end{proposition}
\begin{proof}
See Appendix~\ref{App5}.
\end{proof}

 In Eq. \eqref{eq:MAPdevel} the term $\bm c^\mathsf{T}  \log{ \bm P  \bm u}$ carries the data information from the measurements $\bm c$ while the other two terms carry information from the prior probability distribution $\mathbb{P}(\bm{u}; \bm{\alpha})$. More specifically, the term  $\bm{u}^\mathsf{T}\log{\bm{\alpha}}$ carries information from the prior vector $\bm \alpha$ (or equivalently  $ \bm{a} = C \bm{\alpha}$) which contains the parameters of the prior probability distribution (ref. \ref{sec:remarksinitial}), while 
 the term  $\sum_{j=1}^J{\log{u_j!}} $ 
captures the \emph{combinatorial diffusion} effect that is intrinsic to the adoption of a Multinomial distribution, and induces a preference for solutions with higher entropy (it can be easily shown that this term is maximized when all elements of $\bm u$ are equal). 

The discrete factorial term appearing in Eq. \eqref{eq:MAPdevel} could be replaced by its analytical continuation, i.e., by the Gamma function $\Gamma(n+1)=n! $, leading to an equivalent continuous 
function that, in principle, can be minimized numerically via standard numerical methods (e.g., gradient descent). 
However, for very large problem instances, with $I$ and $J$ in the range of tens of thousands,  numerical resolution with general purpose solvers might still be too impractical, motivating the derivation of a computationally simpler alternative.

In Section \ref{sec:ourestimator} we will present a novel alternative estimator,  labelled DF for ``Data First", 
that exploits the particular structure of the problem at hand and yields a closed-form analytical solution. In order to provide additional insight, we will also elaborate on the relation and conceptual difference between the new DF estimator and the classical MAP estimator derived above. To this aim, it is convenient to derive   an approximated version of the MAP estimator. The approximation relies on the  multivariate normal approximation of the (prior) multinomial distribution, yielding the following result (details can be found in  Appendix~\ref{App6}).

\begin{proposition}\label{prop5}
The MAP estimator in Eq. \eqref{eq:MAPdevel} can be approximated by solving the following simpler optimization problem:
\begin{equation}
\hat{\bm u}_\text{A.M.} = \argmax_{ \substack{ \bm{1}^\mathsf{T}_J \bm u  = \myctot \\  \bm u  \geq \bm 0} }{ \big\{ \bm c^\mathsf{T}  \log{ \bm P  \bm u}  -\frac{1}{2} \| \bm u - \bm a \|^2_{\bm a} \big\} }  
 \label{eq:MAP2normal2}
\end{equation}
where $\| \bm u - \bm a \|^2_{\bm a} \eqdef (\bm u - \bm a)^\mathsf{T} \bm A^{-1}(\bm u - \bm a) = \sum_{j=1}^J\frac{( u_j - a_j)^2}{a_j}
$, with $\bm{A}=\mathrm{diag}(\bm{a})$, is the \emph{weighted} $\ell^2$-norm of the difference vector $\bm u - \bm a$.
\end{proposition}
\begin{proof}
See Appendix~\ref{App6}.
\end{proof}

\subsection{A ``Data First" approach}
\label{sec:ourestimator}

In  Section \ref{sec:newresults} we have established that all points within a non-empty MLBS 
attain the same maximum value of the likelihood function.
In other words, the estimation problem is  structurally non-identifiable,  following the definition given in \cite{raue13}, in the sense  that it is not possible to pick a unique solution 
based solely on the information contained \emph{in the data} $\bm c$. As all points within the MLBS conform to the data equally well. In order to select a particular point within the MLBS, or at least restrict to a subset of preferred solutions, 
we must necessarily resort to some additional assumption and/or to external information. 

One possibility is to demand that, in addition to conforming to the measured data, the solution shall conform also to some particular structural property, characterizing the physical process we aim to measure, e.g., in the form of a smoothness criterion, minimum gradient between adjacent tiles, etc. A second approach is to resort to some kind of \emph{prior information} to aid the solution selection process.
Both approaches amount to adding an additional component to the objective function to be minimized, and therefore may be interpreted as different forms of regularization. 
For our specific application, dealing with human distribution in space, there are no obvious structural constraints tied to the physics of the underlying process, and it appears more natural to encode external information in the form of prior information (as done, e.g., in \cite{benjamin2018}). 

In principle, a possible way to account for prior information (however determined) is through the MAP estimator  derived in the previous section. As discussed, the MAP solution strikes a balance between the information \emph{from the data}, captured by the likelihood term embedding  the measurement vector $\bm c$, with  information \emph{before the data}, captured by the prior distribution embedding the ``prior vector'' $\bm \alpha$ (or equivalently $\bm a$).  Likewise for other Bayesian estimators,  as more data (more samples) are available  the solution component driven by the data increases its relative importance and eventually dominates over the prior information. On the opposite direction,  when the data are scarce, the prior information may dominate. 
In the extreme situation that \emph{only a single  sample measurement is available}, as is specifically the case at hand in our application, adopting a MAP estimator involves a certain risk of diminishing the relative weight of the measurement information to the point that it almost vanishes. In other words, as with MAP (exact or approximated) we cannot control explicitly how much weight to put on the data vis-à-vis the prior, there is a certain risk of  ending up with a solution that reflects mostly the prior information, only slightly perturbed by the measurements. 

To avoid this undesirable effect, in the following we develop an alternative heuristic procedure for the problem at hand based on the ``Data First" principle, where in-data information (likelihood) is given priority over other off-data knowledge (prior distribution and/or any other structural property).

The goal expressed above is achieved though an estimator built according to the following structure:
\begin{equation}
    \hat{\bm u} = \argmin_{\bm u \in \mathcal{U}}{f\big(\bm u, \bm a \big)} 
    \label{eq:mymlestruct}
\end{equation}
with $f(\bm u, \bm{a})$ denoting a distance function between vectors $\bm u$ and $\bm a$.
With this structure, \emph{the ML property is imposed as a hard constraint}. In fact, the condition
\begin{equation}
\bm u \in \mathcal{U} \Leftrightarrow \big\{  \bm P \bm u = \bm c, \; \bm u \geq \bm 0 \big\} 
\label{eq:c2}
\end{equation}
forces
the solution to lie within the MLBS that is by definition the locus of all ML points, i.e., the points that \emph{best conform to the measured data} $\bm c$. 
 Among these points, then we select the point that additionally \emph{best conforms with the off-data knowledge} encoded in the prior vector $\bm a$, through the minimization of some distance function $ f\big(\bm u, \bm a \big)$. Among many possible choices for the distance function,  opting for the  weighted $\ell^2$-norm defined in Proposition \ref{prop5}, i.e.,
 \begin{equation}
 f\big(\bm u, \bm a \big) =  \| \bm u - \bm a \|^2_{\bm a} \eqdef
 (\bm u - \bm a)^\mathsf{T} \bm A^{-1}(\bm u - \bm a)
 \label{eq:f2}
\end{equation}
with $\bm{A} \eqdef \mathrm{diag}(\bm{a})$, allows us to establish a direct connection between the newly proposed estimator and the approximate MAP form derived in Eq.  \eqref{eq:MAP2normal2}, as elaborated below.
Plugging Eq. \eqref{eq:c2} and Eq. \eqref{eq:f2} into the structure
leads to the following novel estimator
\begin{equation}
    \hat{\bm u} =  \argmin_{  \substack{ \bm P \bm u = \bm c \\ \bm u \geq \bm 0}}{ (\bm u \!-\! \bm a)^\mathsf{T} \bm A^{-1}(\bm u \!-\! \bm a)}.
    \label{eq:mymle}
\end{equation}
The solution defined by Eq. \eqref{eq:mymle} is not available in closed-form and may not even exist if MLBS is empty. However, it provides the basis for an heuristic estimator $\hat{\bm{u}}_\text{DF}$,  labeled DF for ``Data First", based on a suitable relaxation of \eqref{eq:mymle}. Remarkably, this estimate can be always computed conveniently in closed-form, even in case of empty MLBS.  

\begin{proposition}\label{prop6} 
A solution for a relaxation of  \eqref{eq:mymle} is given by
 \begin{equation}
    \hat{u}_{j,\text{DF}} =  \check{u}_j \, \sum_{i=1}^I  c_i  \frac{ p_{ij}  }{\sum_{k=1}^J  p_{ik}  \, \check{u}_k } ,\quad j=1,\ldots, J 
\end{equation}
which can be rewritten as
\begin{equation}
     \hat{\bm{u}}_\text{DF} = \check{\bm{Q}} \bm{c}, \qquad 
    \check{\bm{Q}}  \eqdef \mathrm{diag}(\check{\bm{u}}) \bm{P}^\mathsf{T} \mathrm{diag}^{-1} (\bm{P}\check{\bm{u}} ) \label{eq:DF}
\end{equation}
or, alternatively, as
\begin{equation}
 \hat{\bm{u}}_\text{DF} =  \check{\bm{u}} \odot \bm{P}^T (\bm{c} \oslash \bm{P}\check{\bm{u}} )  \label{eq:DFclosed}
\end{equation}
where 
\begin{equation}
  \check{\bm u} \eqdef \max( \bm{A} \bm{P}^\mathsf{T} (\bm{P} \bm{A}\bm{P}^\mathsf{T})^{-1} (\bm{c}-\bm{P}\bm{a})+\bm{a} , \bm{0}) \label{eq:check_u}
    \end{equation}
(maximum intended element-wise), with $\odot$ and  $\oslash$ defined as in Proposition \ref{prop1}.
\end{proposition}
\begin{proof}
See Appendix~\ref{App7}.
\end{proof}
 
 The idea of DF stems from the fact that the structure of the solution to Eq. \eqref{eq:mymle} can be relaxed to make it independent of the Lagrange multipliers, thanks to the sparsity of $\bm{P}$ (details in the proof). By doing so we obtain the intermediate point $\check{\bm u}$ given in Eq. \eqref{eq:check_u}. Such a point, however,  may violate the  total mass constraint  $\bm{1}^\mathsf{T}_J \bm{u}=C$ due to clipping the negative values to satisfy the non-negativity constraint on $\bm{u}$.
In a second step, we obtain $\hat{\bm{u}}_\text{DF}$ by applying  to  $\check{\bm u}$ the transformation in Eq. \eqref{eq:shepp}, which we had already encountered in the EM procedure, which has the effect of  redistributing the mass between non-zero elements of the input vector so as to guarantee the fulfilment of the total mass constraint by the output vector, i.e., $\bm{1}^\mathsf{T}_J \hat{\bm{u}}_\text{DF}=C$.

Notice that, while   MLE/EM is computed through multiple iterations, the DF estimator can be computed directly.
A careful look at Eq. \eqref{eq:check_u} reveals that the heaviest computation part depends only on the model matrix $\bm P$ and on the  prior $\bm a$. In fact, setting $\bm F \eqdef \bm{A} \bm{P}^\mathsf{T} (\bm{P} \bm{A}\bm{P}^\mathsf{T})^{-1}$ and $\bm g \eqdef (\bm{F}\bm{P}- \bm{I}_J)\bm{a}$, Eq. \eqref{eq:check_u} can be rewritten as
\begin{equation}
  \check{\bm u} = \max( \bm F \bm{c} -\bm{g} , \bm{0})
\end{equation}
where the maximum function is intended element-wise.   
Notice that $\bm{F}$ can be efficiently computed  as $\bm{F} = \bm{A}^{\frac{1}{2}} (\bm{A}^{\frac{1}{2}} \bm{P}^\mathsf{T})^\dag$, where $(\cdot)^\dag$ denotes the Moore-Penrose pseudoinverse of the matrix argument.
Moreover, both matrix $\bm F$ and vector $\bm g$ are independent of the measured data $\bm c$. This turns out to be useful in scenarios where multiple estimates must be computed with different measurement vectors but for the same model matrix $\bm P$, corresponding to real-world situations where the radio network coverage pattern, hence cell footprints, may be assumed to remain unchanged between measurement times while mobile phones may have moved, since in such cases $\bm F$ and $\bm g$ must be computed only once.

\section{Numerical results}\label{sec:numres}
In this section we present numerical results comparing the performance of the different estimators for a sample synthetic scenario. Testing on synthetic data has two important benefits: \emph{(i)} it allows one to control explicitly the data generation parameters and therefore to assess the sensitivity of results to said parameters, and \emph{(ii)} it allows one to quantify the absolute estimation error against a known  ``ground truth".

All the numerical results reported in this section where produced by a set of programs developed in R language. The whole code is made available open-source in the form of an online notebook\footnote{The open-source notebook is publicly available from \url{https://r-ramljak.github.io/MNO_mobdensity}} in order to  enable independent replication of results and reuse of all implemented functions in follow-up work by other researchers.

\subsection{Scenario}
We briefly describe here the data generation scenario, referring the interested reader to the open-source notebook referred above for any further details. 
The reference area covers a total of 
1,600 square kilometers and is divided into a regular square grid of 400 $\times$ 400 = 160,000 tiles of size  100 m $\times$ 100 m each. 

The Ground Truth Population (GTP) of mobile phones was generated based on publicly available official census data for the city of Munich, Germany, and its immediate surroundings\footnote{The census data are publicly available from \url{https://www.zensus2011.de/DE/Home/Aktuelles/DemografischeGrunddaten.html}}.  
The total population was reduced by a factor of $1/3$ to mimic the mobile customer basis of a single MNO with that market share.
The resulting GTP distribution is  shown graphically in Fig. \ref{fig:truepop}. 

The radio network topology was generated synthetically based on the \texttt{mobloc} tool \cite{mobloc} developed by M. Tennekes and used already in other studies \cite{cbsnetmob2019,cbs2019,salgado21}. 
The network generation model is completely parametric and its modular implementation allows to conduct empirical analysis of sensitivity to scenario parameters in future follow-up work. 

The  radio network considered in this study was designed in order to mimic the multi-layer nature of real-world  mobile networks \cite{Mishrabook}. In fact, radio access networks are typically deployed in an incremental way, with a first layer of large cells (also called ``umbrella cells" or macro cells) to ensure total coverage, and additional layer(s) of smaller cells deployed subsequently in order to increase capacity in selected areas and/or  fill residual coverage gaps from the previous layer(s)\footnote{A typical mobile network nowadays consists of the superposition of multiple radio access technologies (RAT) operating at different frequencies, including 2G (GSM 900 and GSM 1800), 3G, 4G and prospectively also 5G. The deployment of each new RAT involves the addition of further layers to the overall radio coverage.}. 
Following the same rationale, our synthetic network 
consists of three distinct cell layers, namely ``macro'', ``meso'' and ``micro'' layers, with different  sizes and densities of radio cells. 

For each layer, antennas  are placed according to a semi-regular hexagonal pattern with superimposed random jitter so as to mimic the irregular placement of antenna towers in real-world deployments while still retaining a certain degree of uniformity in tower density.  The latter condition ensures that the spatial distribution of antennas, hence the spatial pattern of radio coverage, remains independent from, and neutral to, the spatial distribution of mobile phones  across the considered region. In other words, we are neither introducing an explicit matching nor an explicit mismatching between cell density and population density --- an aspect that might advantage or disadvantage some estimator vis-à-vis the others.   
In the considered network deployment there are a total of 204 antenna locations --- 27, 156 and 21 respectively for macro, meso and micro cell layer --- placed at the locations shown in Fig. \ref{fig:contours}. 

Every antenna carries a triplet of $120^\circ$ sector cells oriented in different azimuth directions.  The signal propagation model implemented in \texttt{mobloc}  follows a simple geometric model with configurable parameters.  For a generic tile $j$, the received signal strength from a generic cell $i$ is computed at the tile center and from there 
 the so-called  ``signal dominance'' value $s_{ij}$ is derived. If its value  falls below a minimum threshold (set to 0.05 in our simulations) then tile $j$ is considered to fall outside the coverage area of cell $i$. All cells for which the signal dominance value exceeds the minimum threshold ``cover" tile $j$ and  compete to serve the mobile phones therein. If the tile is covered by multiple radio cells, then each mobile phone selects the serving cell independently from other phones  and with probabilities 
that are proportional to the signal dominance values of the competing cells in that tile, as modelled by  Eq. \eqref{eq:s2p}.

In our scenario the antenna and cell parameters are set to layer-specific values for height, power, azimuth orientation, path loss exponent, etc. In Fig. \ref{fig:contours} we show the contour of the cell coverage areas for three sample antennas --- one for each layer. It should be clear from this figure that the  coverage areas of each cell may overlap (i) with  other cells mounted on the same antenna; (ii) with other cells  mounted on other neighboring antennas from the same radio layer; and (iii) with other cells mounted on other neighboring antennas from different radio layers.

\begin{figure}[tb!]
\centering
\includegraphics[trim={20 10 0 0},clip=true,width=0.99\linewidth]{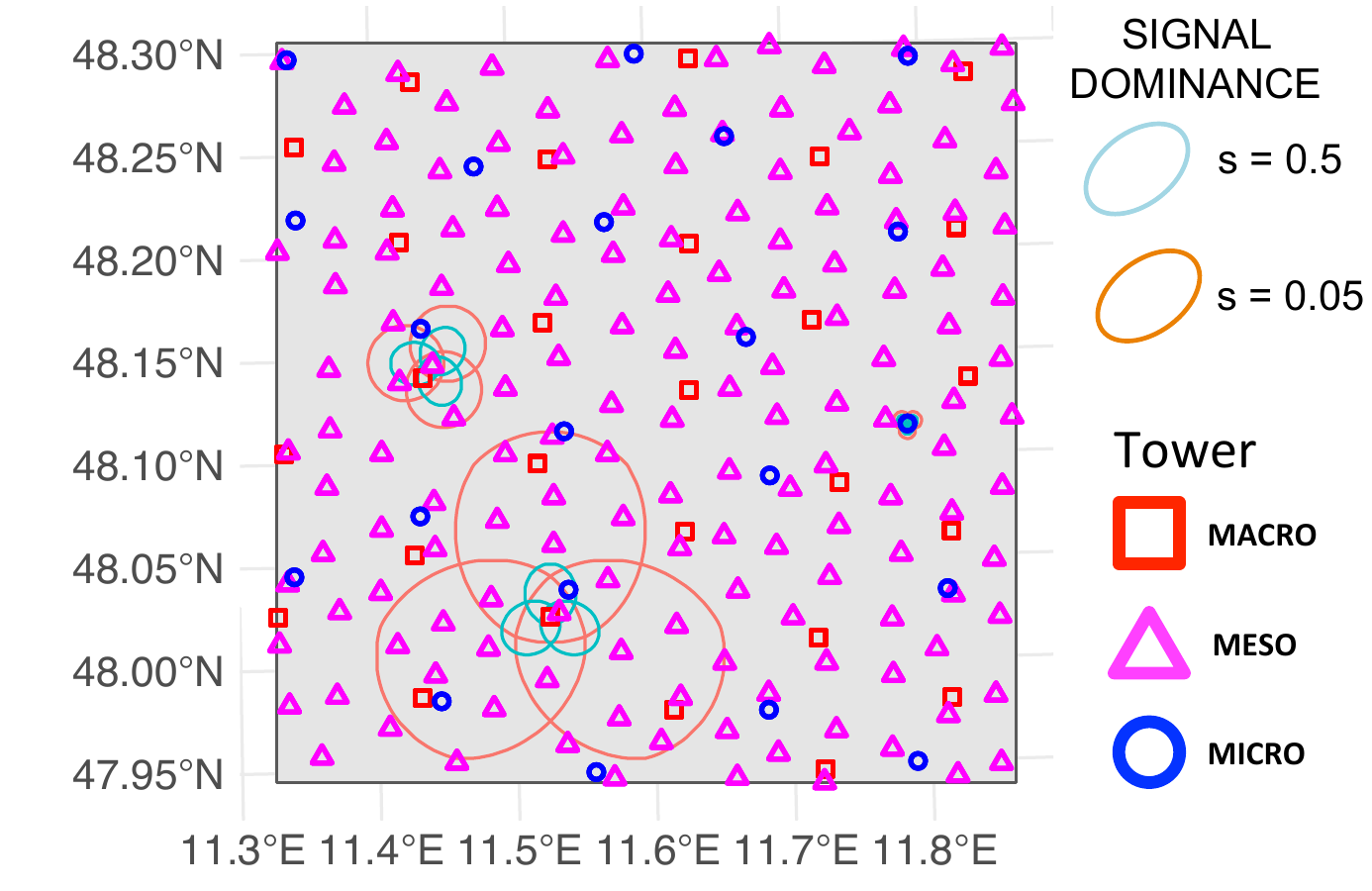}
\caption{Tower locations. Note the different markers for the three radio layers: macro, meso and micro.
The figure  also shows the cell coverage contour plots,  i.e., the locus of points where signal dominance reaches the minimum threshold value 0.05 and the median value 0.5, for one sample tower of each layer.}
\label{fig:contours}
\end{figure}

\begin{figure*}[tb!]
\centering
\subfigure[Vor-T]{\includegraphics[trim={0 10 0 0},clip=true,width=0.32\linewidth]{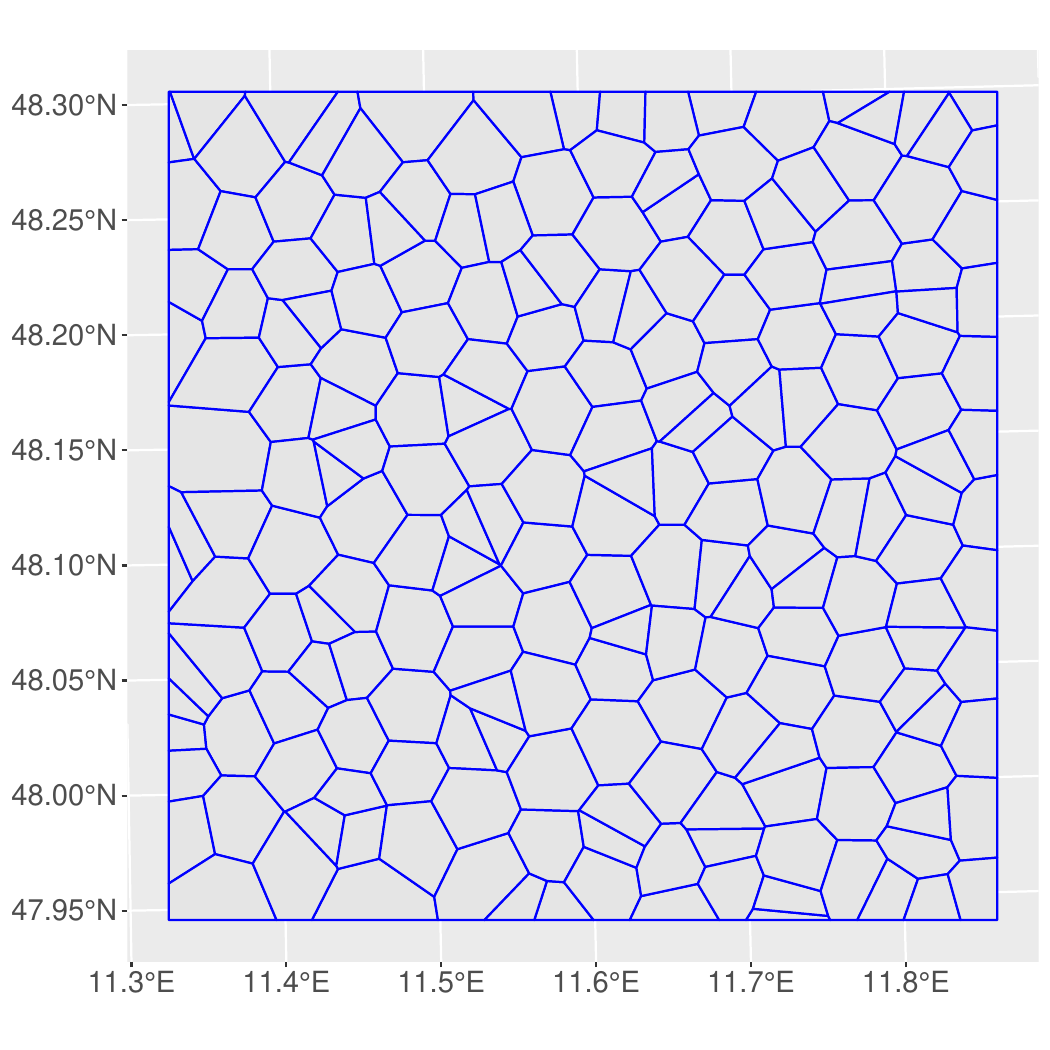}\label{fig:cov_ma}} 
\subfigure[Vor-O]{\includegraphics[trim={0 10 0 0},clip=true,width=0.32\linewidth]{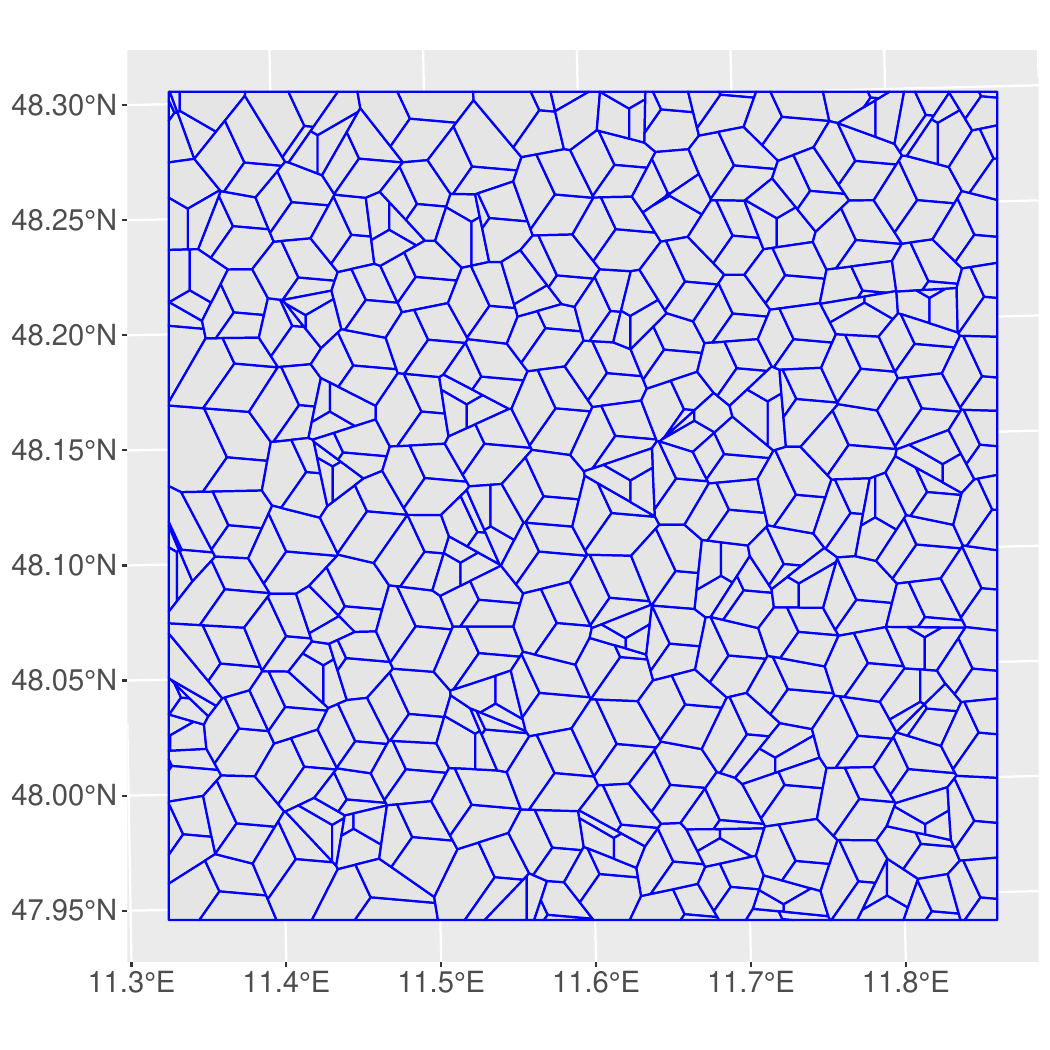}\label{fig:cov_me}}
\subfigure[Vor-B]{\includegraphics[trim={0 10 0 0},clip=true,width=0.32\linewidth]{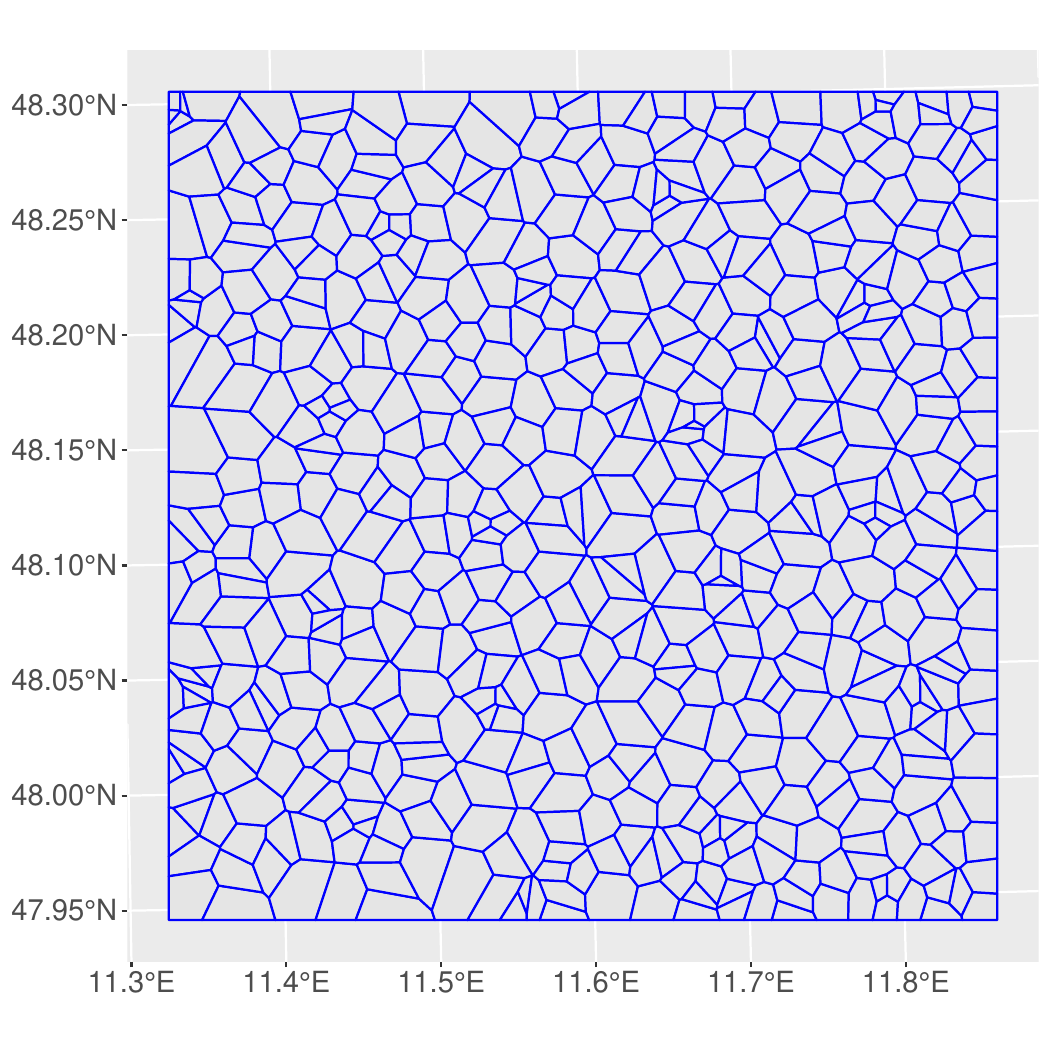}\label{fig:cov_mi}}
\caption{Voronoi diagrams for the three tessellation options.}
\label{fig:voronoidiagrams}
\end{figure*}

\begin{figure*}[tb]
\centering
\subfigure[Ground truth population (GTP)] 
{\includegraphics[trim={20 100 20 80},clip=true,width=0.5\linewidth]{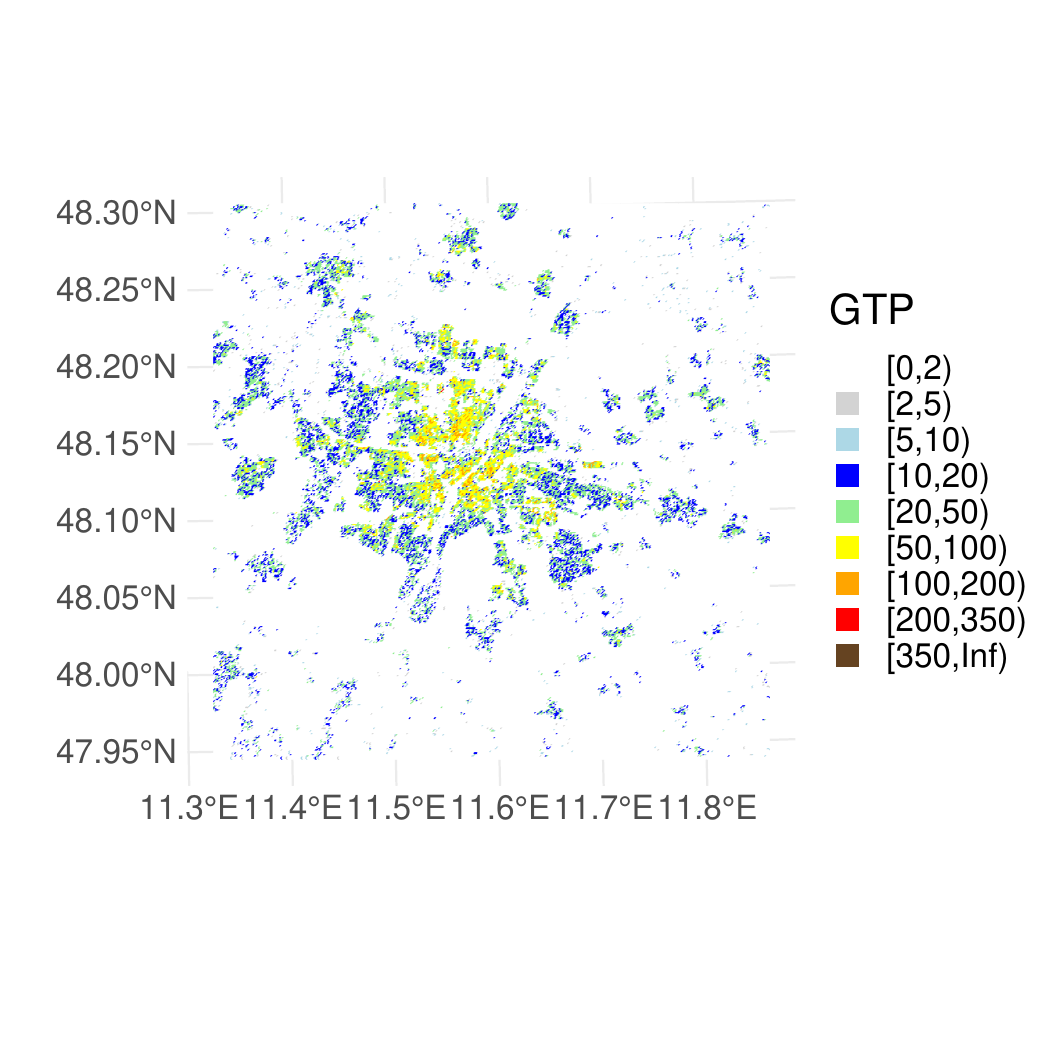}\label{fig:truepop}}\\
\subfigure[Vor-T] {\includegraphics[trim={20 100 120 80},clip=true,width=0.32\linewidth]{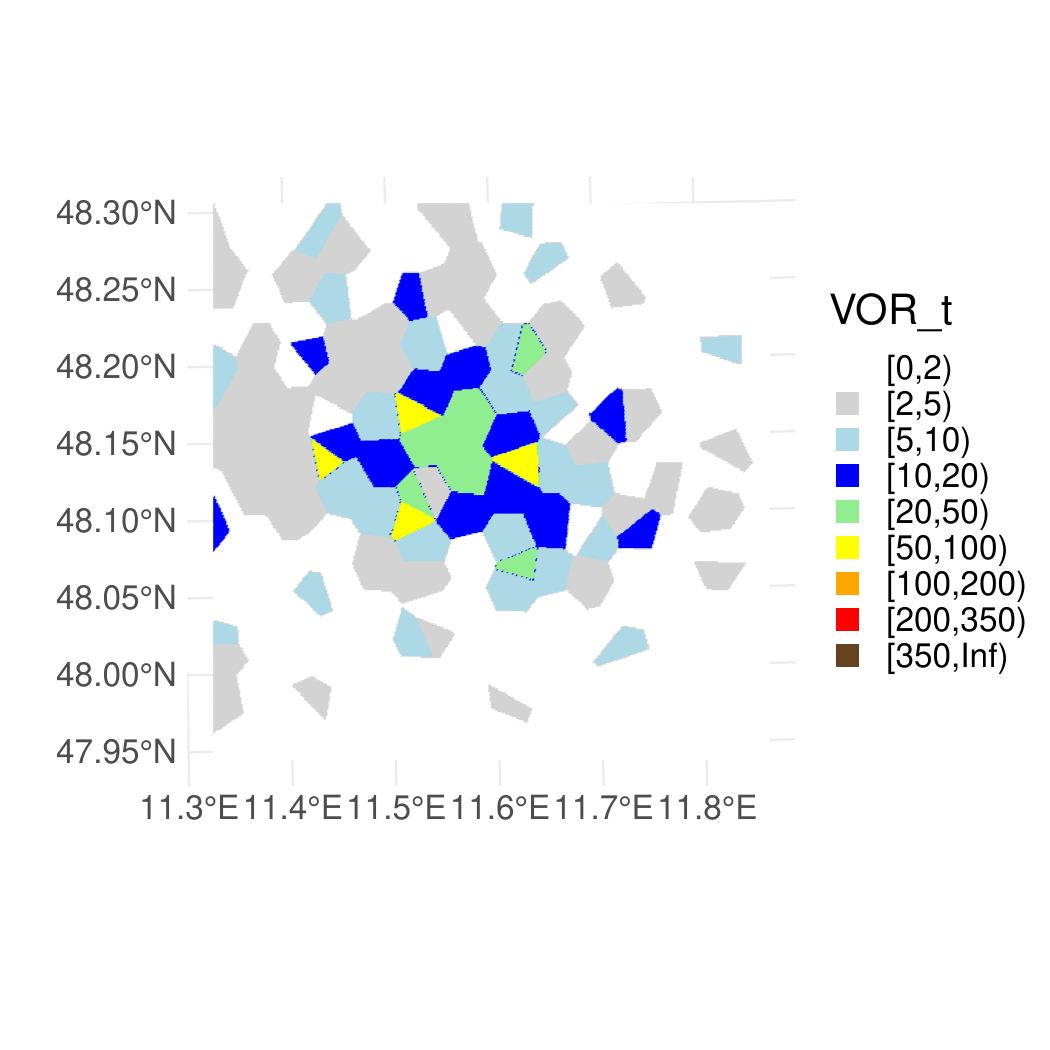}\label{fig:mapvort}}
\subfigure[Vor-O] {\includegraphics[trim={20 100 120 80},clip=true,width=0.32\linewidth]{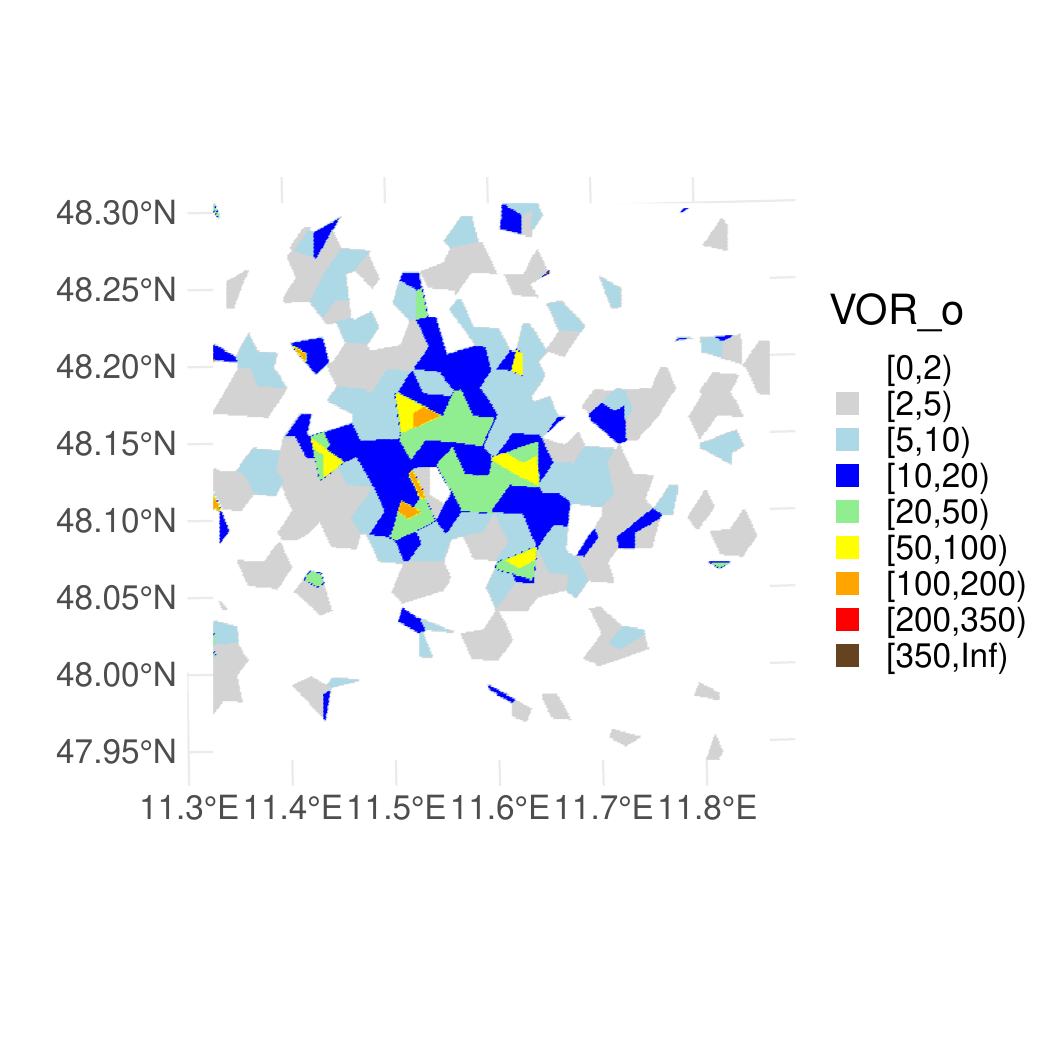}\label{fig:mapvoro}}
\subfigure[Vor-B] {\includegraphics[trim={20 100 120 80},clip=true,width=0.32\linewidth]{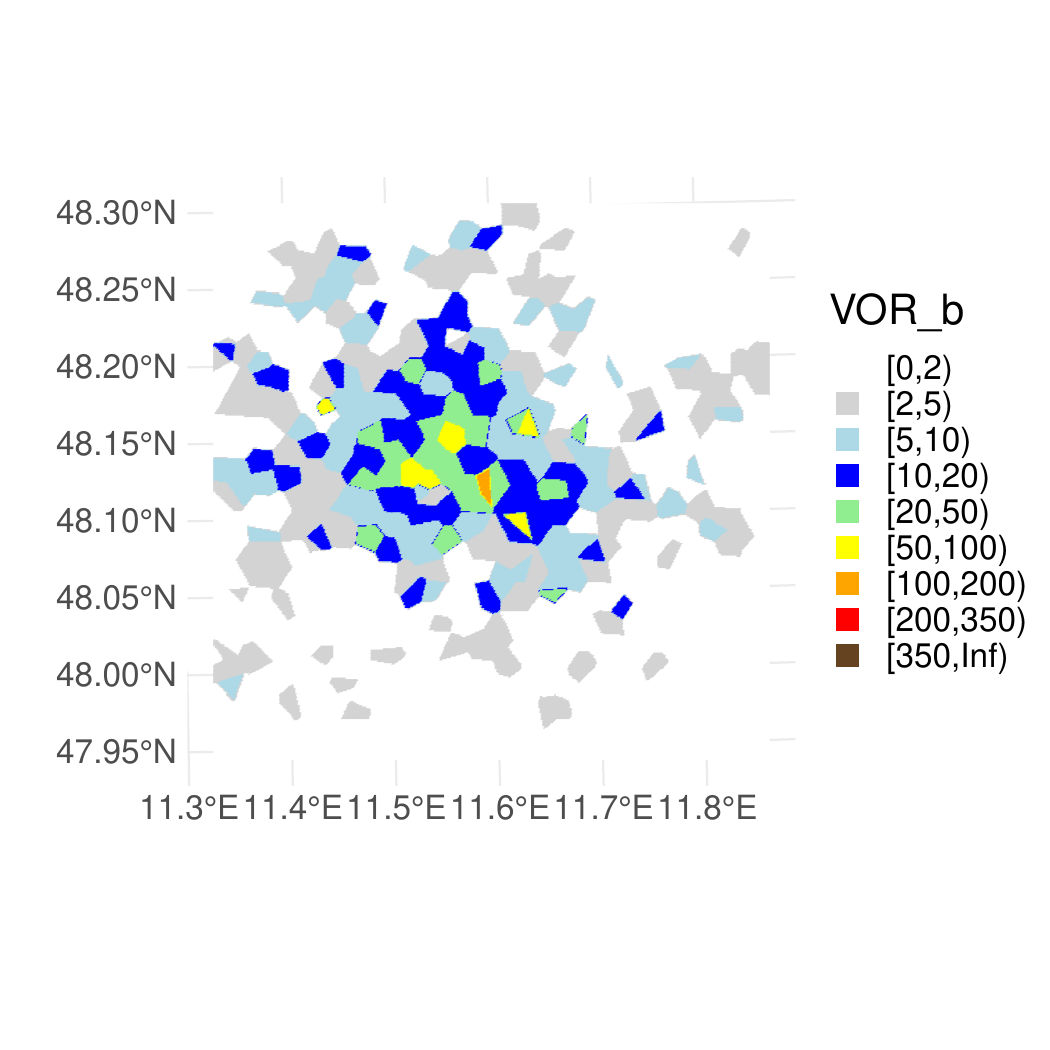}\label{fig:mapvorb}}\\
\subfigure[SB] {\includegraphics[trim={20 100 120 80},clip=true,width=0.32\linewidth]{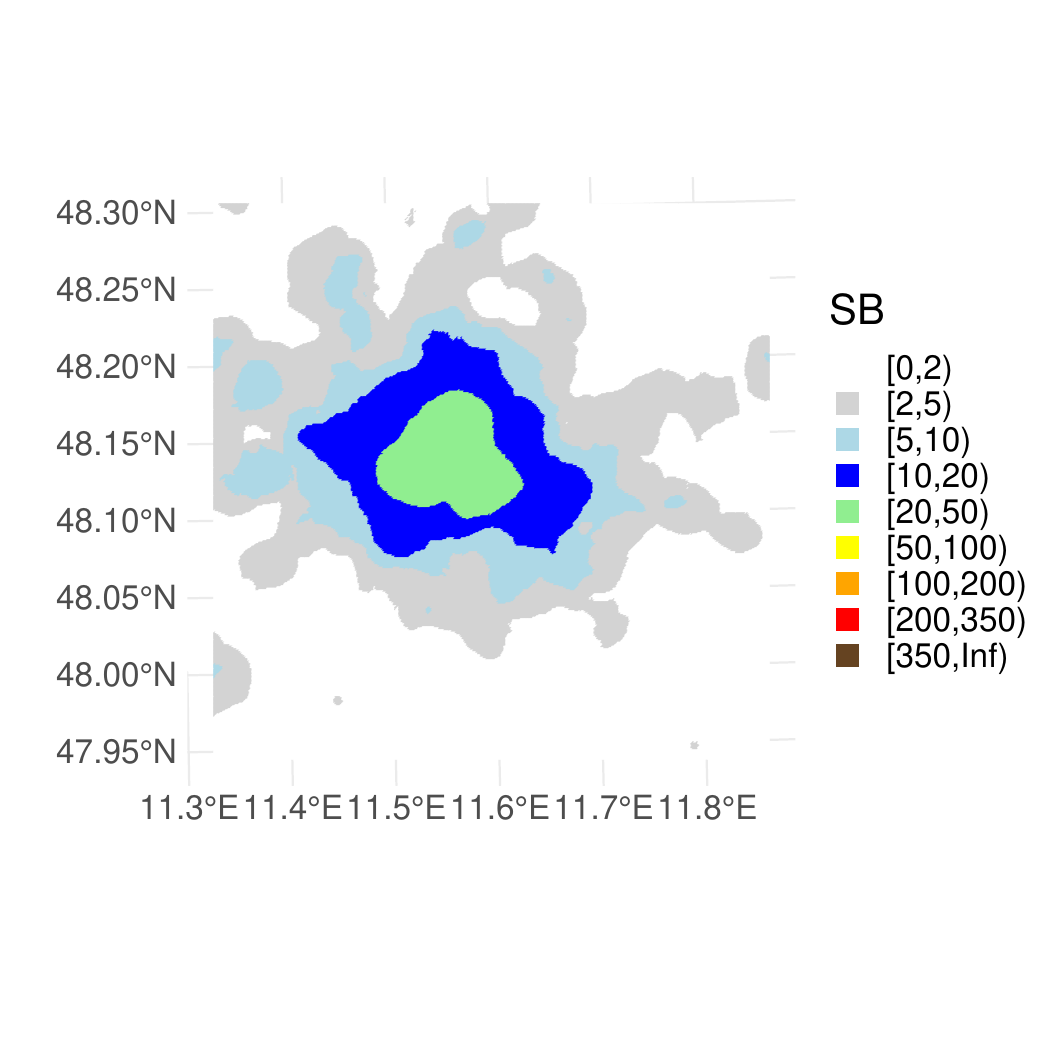}\label{fig:mapsb}}
\subfigure[ML/EM] {\includegraphics[trim={20 100 120 80},clip=true,width=0.32\linewidth]{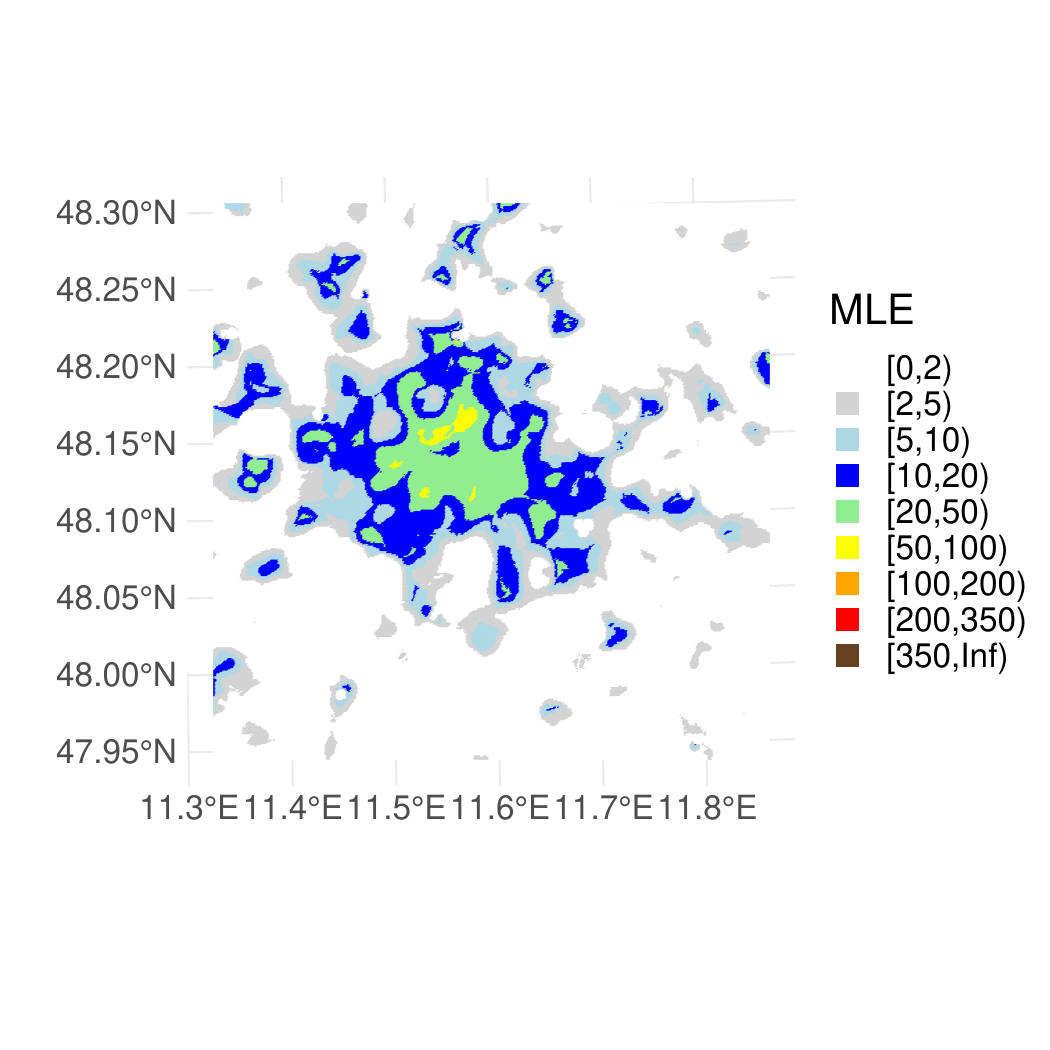}\label{fig:mapmle}}
\subfigure[DF] {\includegraphics[trim={20 100 120 80},clip=true,width=0.32\linewidth]{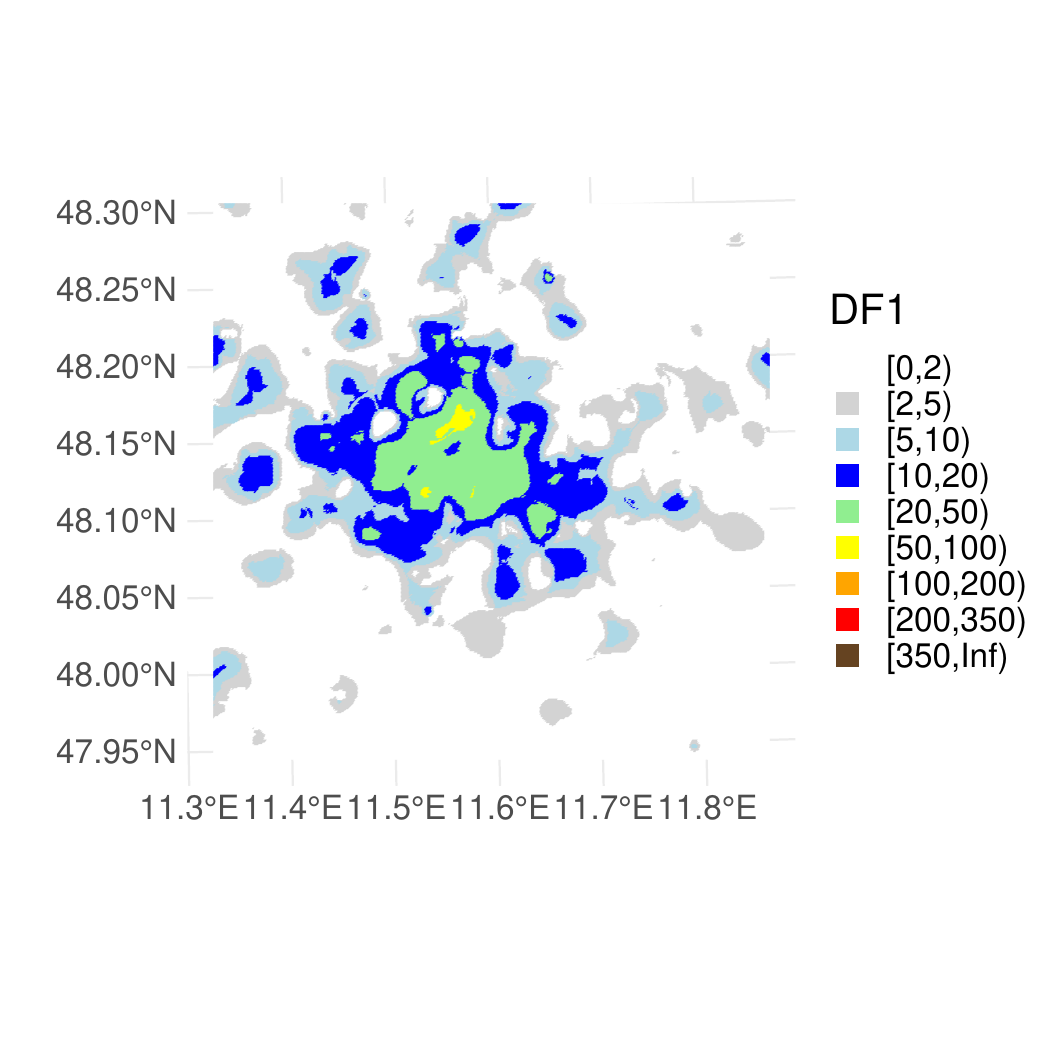}\label{fig:mapdf}}
\caption{Spatial distribution of actual population (a) and corresponding estimates with different methods (b)-(g) plotted in the same color scale.}
\label{fig:populationgrids}
\end{figure*}

\subsection{Estimators}
In our simulation study we have compared the performances of  six different density estimation approaches: three variants of Voronoi tessellations and three different estimators for overlapping cell locations. 

The following Voronoi tessellation methods were considered, resulting in the  Voronoi diagrams shown in Fig. \ref{fig:voronoidiagrams}:
\begin{itemize}
    \item {\bf Vor-T} --  Voronoi tessellation with one seed for each antenna tower. This simple method represents the standard approach in most previous literature. 
    \item {\bf Vor-O} --  Voronoi tessellation with one seed for each radio cell placed at a small fixed distance (10 m) from the respective tower location in the azimuth direction. This method, first proposed by \cite{meersman2016}, is the simplest way to take into account the azimuth orientation of directional cells. 
    \item {\bf Vor-B} --  Voronoi tessellation with one seed for each radio cell placed at the  barycenter (mean point) of the signal dominance profile for the given cell.
\end{itemize}
Furthermore, we have  implemented and compared the following three estimators for overlapping cells.
\begin{itemize}
    \item {\bf SB} -- the simple bayesian estimator in Eq. \eqref{eq:SBelement}. 
     \item {\bf MLE/EM} -- the MLE computed with the iterative EM procedure based on Eq. \eqref{eq:shepp} after $n=200$ iterations.
     \item {\bf DF} -- our novel estimator computed from Eq. \eqref{eq:DFclosed}. 
\end{itemize}
All these estimators are instantiated with a prior vector  $\bm{a}$ and a model matrix $\bm{P}$. For this study we have considered a non-informative uniform prior vector,  with all equal elements. The model matrix was assumed to match exactly the emission matrix used in the data generation process.

\begin{figure}[tb]
\centering
\includegraphics[trim={25 0 0 0},clip=true,width=0.8\linewidth]{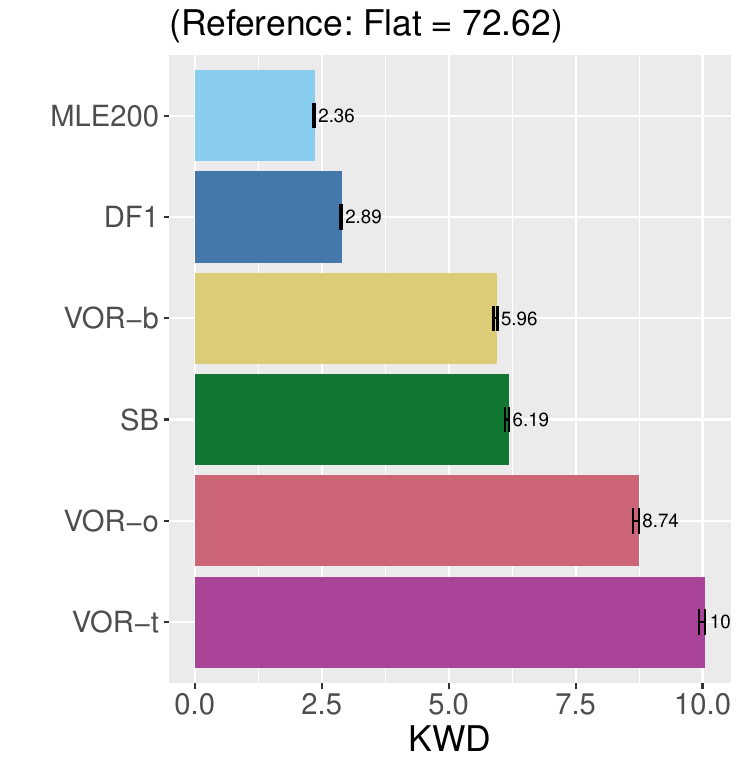}
\caption{Kantorovich-Wasserstein Distances (KWD) between the ground truth density and different estimates. The KWD values are expressed in tile units. As the length of a single tile is 100 me, the KWD value $w=10$ represents an average spatial error of 1 km.}
\label{fig:kwd}
\end{figure}

\begin{figure}[tb]
\centering
\includegraphics[trim={0 120 0 10},clip=true,width=0.99\linewidth]{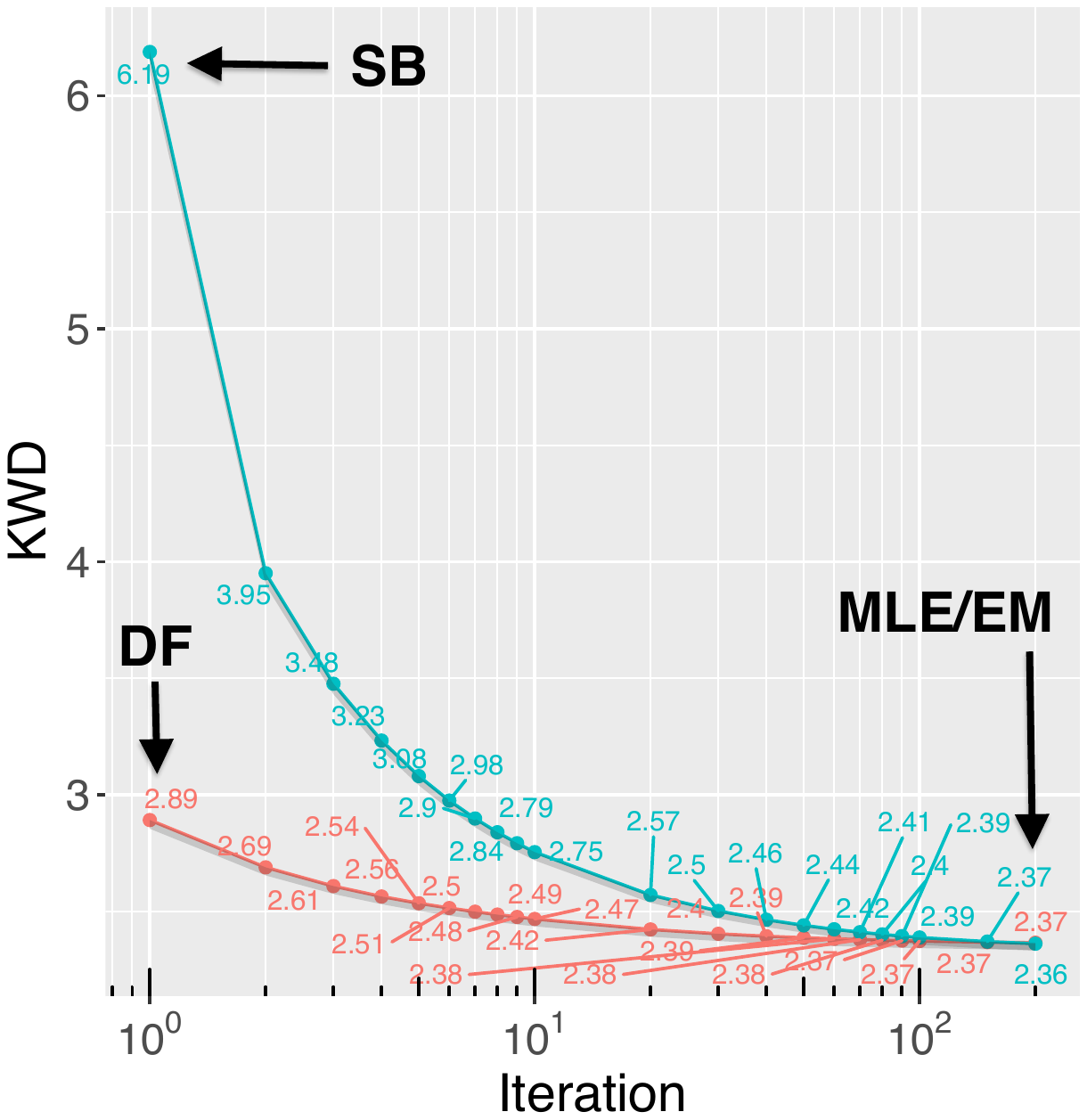}
\caption{Convergence of EM procedure: KWD value (from GTP) against iteration index $n$. }
\label{fig:kwd-convergence}
\end{figure}

\subsection{Similarity measure}
In order to evaluate the performance of a generic estimation method relative to competing candidates we must resort to some measure  of similarity  $w(\hat{\bm u}, \bm v)$ between the estimated density $\hat{\bm u}$  and the ground truth density $\bm v$. To this aim, we resort to the Kantorovich-Wasserstein Distance (KWD) with Euclidean ground distance. KWD (defined formally below) is known in the literature with several different names: Earth Mover Distance, Mallow distance, etc. (see \cite{stegua2019} and references therein). KWD is being increasingly considered in various scientific fields  
as an alternative to  more traditional probabilistic measures like, e.g.,   Kullback-Leibler divergence, Hellinger divergence,  Pearson divergence. These measures belong to the class of \emph{f}-divergences \cite{fdiv10} and are built from the differences between the distribution \emph{values} at individual tiles, with no consideration for their spatial \emph{positions}. Such  tile-by-tile view does not capture the \emph{spatial structure} of the distribution domain, and therefore misses completely the effects of ``horizontal" spatial errors  and  physical proximity (or lack thereof) between distribution masses. Instead, all these aspects are taken  inherently into account by KWD\footnote{The difference between  \emph{f}-divergences and KWD is brightly explained in terms of  ``vertical’’ versus ``horizontal’’ similarity in \url{https://jeremykun.com/2018/03/05/earthmover- distance}.}. 

The limitations of \emph{f}-divergences  are shared also by the  so-called ``Euclidean distance" between distributions, defined as $\| \bm u - \bm v \| = \sum_{j=1}^J{\left(u-v \right)^2}$ ($J$ is the number of elements, i.e., the total number of tiles),  for which the input elements $\bm u$,  $\bm v$ are seen as points in a $J$-dimensional \emph{abstract} space  rather than distributions of size $J$ in a two-dimensional \emph{physical} space, leading to the curious paradox that adopting  \emph{the Euclidean distance between distributions} in the abstract space  effectively disregards \emph{the Euclidean distance between tiles} in the physical space. Instead, the latter is central to the KWD definition and motivates its adoption for the present study. 

For two input distributions $\bm u$ and $\bm v$ defined over the same spatial grid $\mathcal{J}$ and with same total mass $\bm{1}^\mathsf{T}_J \bm u = \bm{1}^\mathsf{T}_J \bm v = \myctot$, KWD may be interpreted as the minimum cost of transporting the mass from configuration $\bm u$ to $\bm v$ (or vice-versa) when the cost of transporting a unit of mass between two generic tiles $j$ and $k$ equals the (physical) ground distance $d_{kj}$ between them. KWD can be expressed as 
the solution $w$ of the following Linear Programming (LP)  optimization problem:
\begin{equation}
\begin{array}{ll@{\quad}ll}
\text{minimize}  & w(\bm u, \bm v) \eqdef \frac{1}{C}\sum_{k}^{}{\sum_j { x_{kj}d_{kj} } } & & \\
\text{subject to} & \sum_{j}{x_{kj}}  = u_k   &   \forall \, k  \in \mathcal{J} \\
			& \sum_{k}{x_{kj}} = v_j   &   \forall \, j \in \mathcal{J} \\
                 &                                              x_{kj} \geq 0  & \forall \, k,j  \in \mathcal{J}
\end{array}
\label{eq:LPGall}
\end{equation}
It can be easily shown that KWD defined in this is symmetric and fulfills the  triangular inequality, therefore qualifies fully as a ``distance" in the formal mathematical sense. 

The direct resolution of the LP problem in Eq. \eqref{eq:LPGall} is computationally expensive, preventing the computation of exact KWD values on very large grids. 
In a recent work Gualandi \emph{et al.} \cite{stegua2019}  showed that a close approximation to the exact KWD value, within a provable deterministic bound, can be obtained by solving a transportation problem over a regular lattice. In the proposed solution, the integer parameter $L$ determines  the lattice density and acts as a tuning knob to trade-off computation resources (time and memory) with approximation error: they show that $L=3$ is sufficient to achieve a KWD approximation that is guaranteed to be within 1.2\% of the exact KWD value in all scenarios.  An efficient open-source implementation of their method is publicly available along with code wrappers for R and Python\footnote{ \url{https://github.com/eurostat/Spatial-KWD}}. All KWD values presented in the present study were obtained with the R package
 \texttt{SpatialKWD} available in CRAN\footnote{\url{http://cran.r-project.org/web/packages/SpatialKWD}}. Unless differently specified we retained the default parameter setting $L=3$.

A physical interpretation of KWD $w(\hat{\bm u}, \bm v)$  in our context is the following. Imagine that, starting from 
the estimated density $\hat{\bm u}$, we move each unit of mass (or mobile phone) across tiles so as to arrive at the final GTP configuration $\bm v$ (transportation plan), and we do so in a way that minimizes the total travelled distance (\emph{optimal} transportation plan). Some  units will not travel, while others will travel over shorter or longer distances. \emph{On average, a single unit will travel  a distance equal to the KWD value $w$}. 
In other words, since KWD is normalized to the total mass (note the factor $\frac{1}{\myctot}$ in Eq.  \eqref{eq:LPGall}) its value can be interpreted as the ``average distance'' in the optimal transportation plan between GTP and the estimated map (or vice-versa), and therefore as the ``average spatial error'' of the estimated density against the ground truth. 

\subsection{Results}
In Fig. \ref{fig:populationgrids} we report the density maps obtained with the six different estimation methods along with the ground truth population.
The corresponding KWD values between the final estimates and  GTP are reported  in Fig. \ref{fig:kwd}. 

Recall that for SB, ML/EM and DF the prior vector was set to the non-informative uniform vector with all equal elements (flat prior).  In the considered scenario the KWD value between the flat prior and GTP was 72.6 (average error of 7.6 km). This represents a naive absolute reference for assessing the accuracy gain of the estimation process. Even the simplest Vor-T estimator brings the KWD down to 10.0 (average error of 1 km), with a reduction factor of $\times 7$ from the flat prior reference. This confirms the (not entirely obvious) expectation that even the simplest density estimation approach is ``better than nothing". 

If cell azimuth orientation data is available, in addition to tower locations, one can implement the slightly more sophisticated Vor-O variant. This brings the KWD score down to 8.7 (average error 870 m), i.e., a 15\% improvement upon the simpler Vor-T option. 

If further detailed knowledge about the signal dominance values is available, then one could use this information within a Voronoi-based approach by setting the seeds to the exact cell barycenter. This is, in a nutshell, the Vor-B variant. One would expect Vor-B to perform better than Vor-O, since it uses more detailed information about the cell coverage area (beyond merely tower locations and azimuth orientation) and in fact the KWD goes further down to around 6 (average error 600 m).   

But if such detailed information is available, then one could also resort to the other three estimators that were designed specifically for overlapping cells. While the simplest SB estimator performs comparably to Vor-B in the considered scenario, a more substantial improvement can be obtained with DF and MLE: both methods bring the KWD below 3 (average error below 300 m) with a slight advantage of MLE/EM that yields  an average error around 240 m, i.e., a factor of $\times 4$ improvement over the most popular Vor-T method.

The price to pay for the higher accuracy of DF and MLE/EM vis-à-vis the other schemes  is a slightly higher computational complexity. However, for both DF and MLE/EM the computational burden is absolutely sustainable.  
Our naive R implementation took less than 4 seconds to compute the closed-form DF estimate for the considered scenario on a low-end laptop, and further optimizations are possible. 
As for the iterative EM procedure, 
our tests show that the number of iterations required to reach a good solution is not very high, in the order of tens or a few hundreds (see Fig. \ref{fig:kwd-convergence}) and, most importantly, each iteration is relatively straightforward to compute due to the simple structure of Eq. \eqref{eq:shepp}. 
Our naive, single-threaded implementation took less than 1 second  to complete one iteration for all 160,000 tiles in the considered scenario.
Since Eq. \eqref{eq:shepp} must be evaluated once for every super-tile at each iteration, the overall computation time scales linearly with the the number of 
super-tiles\footnote{It should be noted that the number of super-tiles scales sub-linearly with the size of the area of interest. For example, when the territory of a whole country is considered, large rural areas falling under the coverage of a single cell tend to be aggregated into a single super-tile. In general, rural areas with sparse radio infrastructure tend to contribute with a lower number of super-tile than densely populated urban areas of the same size.}. 
With further code optimization, parallelization and additional hardware resources the computational cost of MLE/EM is acceptable even for much larger instances. 

Finally, we provide some insight into the convergence behaviour of the iterative MLE/EM procedure.  Fig. \ref{fig:kwd-convergence} plots the KWD (against GTP) of the output of the MLE/EM procedure after $n$ iterations for different values of $n$ (note the logarithmic scale) when the starting point is set to the non-uniform flat prior. Recall from Proposition \ref{propSBfirstiteration}  that the outcome after the first iteration ($n=1$) is equivalent to SB. One additional iteration  ($n=2$) is sufficient to halve the average error. The speed of improvement slows down and basically stabilizes after the first $n=100$ iterations. 
Fig. \ref{fig:kwd-convergence}  also reports  the value obtained with the DF estimator from Eq.  \eqref{eq:DFclosed}. Recall that the DF solution can be computed directly, with no iteration. However, the DF solution can also be  used as an alternative initialization point for the iterative EM procedure. During the study we wondered whether such a different initialization strategy could eventually lead to a better final solution, and for this reason we report in Fig. \ref{fig:kwd-convergence} also the KWD values obtained by the iterative EM procedure when the initial point is set to the DF solution. However, we found that the solution accuracy after $n=200$ iterations remains basically unchanged in the considered scenario.

\section{Conclusions and future work}\label{sec:conclusions}

The analytical and numerical results presented in this study indicate that, when the model matrix $\bm P$ is known accurately, the best estimates are obtained with MLE/EM and DF and that the accuracy gain with respect to the more popular Voronoi method is substantial, by a factor of $\times 4$ for Vor-T (based on tower location data) and $\times 3$ for Vor-O (considering also cell direction data) when quantified through KWD.
A substantial improvement, by a factor of $\times 2.5$, is also in place against the SB estimator. 

The gain in accuracy of the MLE/EM and DF estimators comes at the price of a relatively contained computation cost, thanks to the relatively simple structure of the iterative  Eq. \eqref{eq:shepp} for MLE/EM and the availability of a closed-form solution for DF. Our simple single-threaded implementation solves an instance of  160,000 tiles on a low-end commercial laptop in about 3 minutes for MLE/EM with $n=200$ iterations, and less than 4 seconds with DF. Early tests on very large instances up to 2,000,000 tiles (not reported in this study) indicate that computational cost would not be a critical factor for a professional implementation. 

 We remark that the MLE/EM solution is analytically equivalent to the MLE solution developed independently in \cite{ricciato16}, but the resolution procedure derived in Eq. \eqref{eq:shepp}  is much simpler to implement and faster to compute than resorting to general-purpose solvers as done in \cite{ricciato16}. 

 The main limitation of all Voronoi approaches is the implicit assumption, inherent to the adoption of a tessellation approach,  that mobile phones always connect to the nearest cell tower. This assumption is problematic in  real-world scenarios due to several features of cellular radio planning: multi-layer radio coverage design, mixing of small/large cells, dynamic power control and load balancing mechanisms, etc. On the other hand, the simplest Voronoi variants Vor-T and Vor-O require very little information about the radio network, and therefore remain appealing for practical applications that can tolerate lower levels of spatial accuracy. 
 
Conversely, if spatial accuracy is at premium, efforts should be invested in acquiring more detailed information about radio network coverage in the form of detailed cell footprint profiles $s_{ij}$'s, based on which more accurate predictions of emission probabilities $p_{ij}$'s can be derived.

 In between the ``low-information low-accuracy" methods, namely Vor-T and Vor-O,  and ``high-information  high-accuracy" methods, namely  MLE/EM and DF,  our results  indicate that little room is left for intermediate approaches like Vor-B or SB. The latter seem to inherit the weaknesses of the other two classes:  they require very detailed information in order to be instantiated, comparable to that needed for MLE/EM and DF, but they fail to make good use of such information and end up with a level of accuracy that is only slightly better than the much simpler Vor-T and Vor-O methods.

The above results contribute to advance the understanding of the relative advantages and costs of the recently proposed geolocation methods based on overlapping cells  vis-\`a-vis the more traditional tessellation methods. At the end of the day, in all practical applications the choice between one approach or the other is a matter of balancing benefits and costs.  
On the cost side, based on the present study we can claim  that the dominant factor is not related to computation resources (and time) but rather to the efforts required to acquire detailed data about cell coverage profiles.
On the benefit side, our simulation framework allows to assess the maximum possible gain that can be expected for different network deployment scenarios under the ideal assumption that the emission probability matrix $\bm P$ is known precisely. 

The natural next step in the research is to investigate the sensitivity of the final solution to uncertainty in model parameters $p_{ij}$'s, or in other words assess the robustness of  estimators to model mismatching errors.
Another distinct but closely related direction of research should look at assessing the  impact of different scenario parameters on the relative performances of the various estimators.
Finally, another direction for further exploration
concerns the effect of informative priors, both in terms of potential gain or loss of accuracy that may be yielded by considering  correct or incorrect informative priors, respectively.
Answering these open research questions is part of our planned follow-up work.

\section*{Acknowledgments}
The authors are grateful to Marco Ramljak for his help with the implementation or R routines for all estimators and data generation modules, and for making his code publicly available in the open-source notebook available at \url{https://r-ramljak.github.io/MNO_mobdensity}. The authors thank the three anonymous reviewers for their constructive and very detailed comments which helped to improve an earlier version of this work.  

\bibliographystyle{unsrt}
\bibliography{densitybiblio2}

\begin{IEEEbiography}
 [{\includegraphics[width=1in,height=1.25in,clip,keepaspectratio]{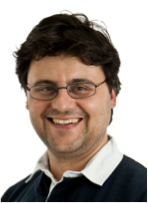}}]
{Fabio Ricciato} 
Fabio Ricciato received a Laurea in Electrical Engineering (1999) and PhD in Information and Communication Technologies (2003) from University La Sapienza, Italy.
Between 2004 and 2017 he worked in telecommunications research in different institutions. He served as assistant professor in Italy (Faculty of Engineering at University of Salento) and associate professor in Slovenia (Faculty of Computer Science at University of Ljubljana), teaching different subjects in the telecommunication fields from networking to signal processing. He served as middle manager in two private Research \& Technology Organisations in Vienna, namely the Telecommunications Research Center Vienna and the Austrian Institute of Technology, leading small and medium size research units of up to 45 researchers. 
In January 2018 he joined the European Commission as a member of the Eurostat Task Force on Big Data. He is currently serving as Statistical Officer in the Eurostat Unit A5 on Methodology and Innovation in Official Statistics.  His current interests revolve around the modernisation of official statistics, exploitation of novel data sources, privacy-preserving approaches to cross-institutional data analytics and analysis of mobile network operator data. 
\end{IEEEbiography}

\begin{IEEEbiography}
 [{\includegraphics[width=1in,height=1.25in,clip,keepaspectratio]{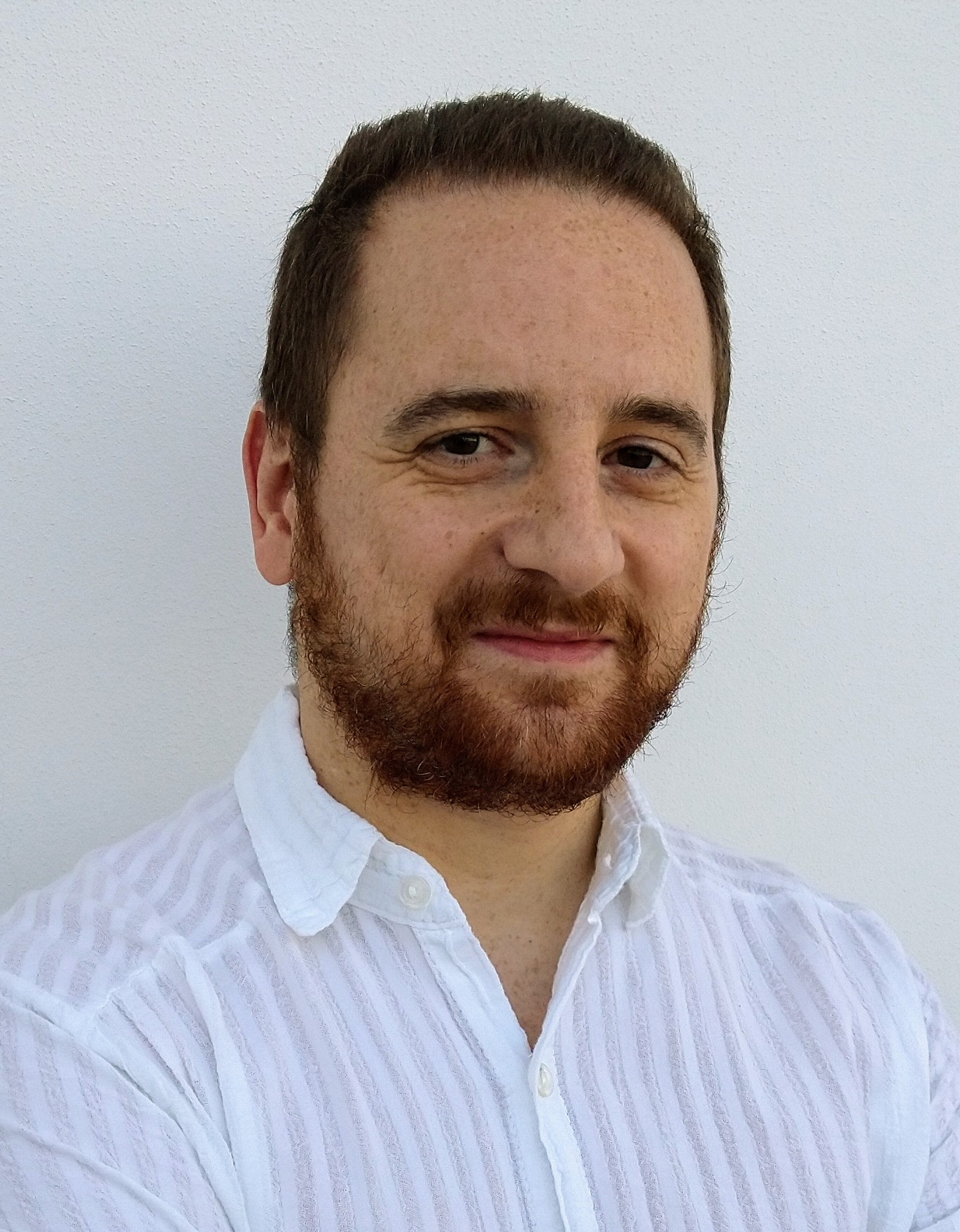}}]
{Angelo Coluccia} 
(M'13-SM'16) received the PhD degree in Information Engineering in 2011, and is currently an Associate Professor of Telecommunications at the Department of Engineering, University of Salento (Lecce, Italy). He has been a research fellow at Forschungszentrum Telekommunikation Wien (Vienna, Austria), and has held a visiting position at the Department of Electronics, Optronics, and Signals of the Institut Sup\'erieur de l’A\'eronautique et de l’Espace (ISAE-Supaero, Toulouse, France). His research interests are in the area of multi-channel, multi-sensor, and multi-agent statistical signal processing for detection, estimation, localization, and learning problems. Relevant application fields are radar, wireless networks (including 5G and beyond), and emerging network contexts (including intelligent cyber-physical systems, smart devices, and social networks). He is Senior Member of IEEE and Member of the Technical Area Committee in Signal Processing for Multisensor Systems of EURASIP (European Association for Signal Processing).
\end{IEEEbiography}

\cleardoublepage

\appendix
\section{Supplemental Material (Proofs)}

\subsection{Proof of Proposition~\ref{prop_convex}}\label{App2}

It is  sufficient to show the negative semidefiniteness of the Hessian of the log-likelihood given in Eq. \eqref{eq:loglike0}, which implies that it is a concave function. The generic  $(k,h)$ element of the Hessian is given by the second derivatives, i.e.,
\begin{equation}
  H_{k,h} = - \sum_{i=1}^I c_i \frac{p_{ik} p_{ih}}{(\bm{p}_i^\mathsf{T} \bm{u})^2}  
\end{equation}
wherein $\bm{p}_i^\mathsf{T}$ denotes the $i$th row of $\bm{P}$. Thus, for any vector $\bm{x}=[x_1 \cdots x_I]^\mathsf{T}$ and assuming $c_i >0$,\footnote{Possible values  $c_i=0$ do not count for the sake of the argument, since they do not impact the definiteness of the Hessian.} by recalling that the elements of $\bm{P}$ are non-negative (of which at least one is strictly positive for each column), we can write
\begin{align}
\sum_{k,h=1}^J x_k x_h H_{k,h} &= - \sum_{k,h=1}^J x_k x_h \sum_{i=1}^I c_i \frac{p_{ik} p_{ih}}{(\bm{p}_i^\mathsf{T} \bm{u})^2} \nonumber\\
&= - \sum_{i=1}^I  \frac{c_i}{(\bm{p}_i^\mathsf{T} \bm{u})^2} \left( \sum_{k=1}^J  p_{ik} x_k \right)  \left( \sum_{h=1}^J p_{ih} x_h \right) \nonumber\\
&= - \sum_{i=1}^I y_i^2
\end{align}
where
\begin{equation}
    y_i = \frac{\sqrt{c_i}}{\bm{p}_i^\mathsf{T} \bm{u}} \sum_{h=1}^J p_{ih} x_h = \sqrt{c_i} \frac{\bm{p}_i^\mathsf{T}\bm{x}}{\bm{p}_i^\mathsf{T} \bm{u}}. 
\end{equation}
It follows that $\bm{x}^\mathsf{T} \bm{H} \bm{x} \leq 0$, with equality only if $\bm{y}=[y_1 \cdots y_I]^\mathsf{T} = \bm{0}$. We can rewrite $\bm{y}=\bm{D}\bm{P}\bm{x}$, where $\bm{D}=\mathrm{diag}([\frac{\sqrt{c_1}}{\bm{p}_1^\mathsf{T} \bm{u}} \ \cdots \ \frac{\sqrt{c_I}}{\bm{p}_I^\mathsf{T} \bm{u}}])$. Since $\bm{D}$ has full rank $I$, $\bm{D}\bm{P}$ has the same rank of $\bm{P}$; the latter has a rank that cannot exceed $I<J$, hence the rank of $\bm{D}\bm{P}$ is smaller than $J$. We conclude that there exists some $\bm{x}\neq \bm{0}$ such that $\bm{y}=\bm{D}\bm{P}\bm{x} = \bm{0}$, hence $\bm{H}$ is negative semidefinite, proving that all stationary points are equivalent global maxima.

\subsection{Proof of Proposition~\ref{prop_optimality}}\label{App3}

At convergence, the stable points of Eq.  \eqref{eq:shepp} must fulfill the condition $\hat{\bm u}^{m+1}=\hat{\bm u}^{m}$ ($m$ denoting the iteration index) therefore for the solution vector $\hat{\bm u}$  it holds:
\begin{equation}
  \sum_{i=1}^I{  c_i \frac{p_{ij} }{ \sum_{k=1}^J{ p_{ik} \hat{u}_k }} } =     \sum_{i=1}^I c_i \frac{p_{ij}}{\bm{p}_i^\mathsf{T} \hat{\bm{u}}} = 1, \quad  \forall j
\label{eq:sheppstable}
\end{equation}
wherein $\bm{p}_i^\mathsf{T}$ denotes the $i$th row of $\bm{P}$. 

For the MLE-Multinomial formula, from Proposition \ref{prop_convex} we know that all stationary points are equivalent global maxima. To identify the stationary points, we derive  the Lagrangian function, i.e., the derivative with respect to $u_j$ of the objective function from Eq. \eqref{eq:MLfinal_u} augmented by the scaled version of the equality constraint. By exploiting the fact that the $i$th element of  vector $\bm{P}\bm{u}$ is $\bm{p}_i^\mathsf{T} \bm{u}$, 
we obtain, for each $j$:
\begin{equation}
\frac{\partial (\bm c^\mathsf{T}  \log{ \bm P  \bm u} - \lambda (\bm{1}^\mathsf{T}_J\bm{u} -C )) }{\partial u_j} = \sum_{i=1}^I c_i \frac{p_{ij}}{\bm{p}_i^\mathsf{T} \bm{u}} - \lambda =0. \label{eq:optimality2}
\end{equation}
Multiplying by $u_j$ and summing over $j$ we obtain
\begin{equation}
\sum_{i=1}^I c_i \frac{\sum_{j} p_{ij} u_j}{\bm{p}_i^\mathsf{T} \bm{u}} = \lambda \myctot
\end{equation}
and since $\sum_{j} p_{ij} u_j = \bm{p}_i^\mathsf{T} \bm{u}$ it turns out that
$\lambda = 1$. By substituting back in the first-order optimality condition Eq. \eqref{eq:optimality2} it follows that, for each $j$,
\begin{equation}
    \sum_{i=1}^I c_i \frac{p_{ij}}{\bm{p}_i^\mathsf{T} \bm{u}} = 1.
\label{eq:supstable2}
\end{equation}
It is evident that the optimality condition derived for MLE-Multinomial in Eq. \eqref{eq:supstable2}  is identical to the convergence condition  of the EM procedure  derived for MLE-Poisson in Eq. \eqref{eq:sheppstable}.
That means, the two methods share  the same solution space\footnote{Strictly speaking, this holds true as long as the initialization vector in the iterative procedure in Eq. \eqref{eq:shepp} is free from zero elements, as discussed in Section~\ref{sec:mlepois}.}.

\subsection{ Proof of Proposition~\ref{prop3}}\label{App4}

The condition $\bm{P}\bm{u} = \bm{c}$ can be rewritten as  $\bm{p}_i^\mathsf{T} \bm{u} = c_i$, $\forall i$, and that implies that the optimality condition from Eq.
\eqref{eq:supstable2}, or equivalently Eq. \eqref{eq:sheppstable}, is always verified: with the {\em ansatz} $\bm{p}_i^\mathsf{T} \bm{u} = c_i$, $\forall i$,  Eq.  \eqref{eq:supstable}  reduces to the identity $\sum_{i=1}^I  p_{ij} = 1$ that is always satisfied by the column-stochasticity of $\bm{P}$, i.e., $\bm{1}_I^\mathsf{T}\bm{P}=\bm{1}_J$. 
Therefore, we have established that $\mathcal{U}$, if non-empty, represents a set of  solutions for both MLE procedures, i.e., any point in $\mathcal{U}$ is also a solution of MLE-Multinomial and MLE-Poisson. 

\subsection{Proof of Proposition~\ref{prop_MAP}}\label{App5}

Following the stochastic  model shown in Fig. \ref{fig:hierarchicalmultinomial},  $\bm u$  represents a random vector with Multinomial distribution of unknown parameter $\bm \mu$. Taking the ``prior vector'' $\bm \alpha$ as the prior value (before seeing the data) of  $\bm \mu$, we obtain that the ``prior probability distribution'' of $\bm u$, denoted by $\mathbb{P}(\bm{u}; \bm{\alpha})$, is   Multinomial with parameters $\myctot$ and $\bm \alpha$, formally
\begin{equation}
\mathbb{P}(\bm{u}; \bm{\alpha}) = 
\mathcal{M}\left(\myctot, \bm{\alpha} \right) = 
\frac{\myctot!}{u_1! u_2! \cdots u_J!} \prod_{j=1}^{J}{ \left( \bm \alpha_j \right)^{u_j}}
\end{equation}
from where the logarithm 
\begin{equation}
\log{\mathbb{P}(\bm{u}; \bm{\alpha})} =  \bm{u}^\mathsf{T}\log{\bm{\alpha}} -  \sum_{j=1}^J{\log{u_j!}}
\end{equation}
is obtained (up to a constant additive term).

The posterior probability distribution is obtained by the product of likelihood and prior probability distribution $\mathbb{P}(\bm{u}; \bm{\alpha})$.
Accordingly, the logarithm of the posterior distribution equals the sum of the log-likelihood given in Eq. \eqref{eq:loglike0} and $\log \mathbb{P}(\bm{u}; \bm{\alpha})$. Therefore, maximizing the (logarithm of) posterior distribution leads to maximizing the following function of $\bm \mu $ and $\bm u$  
\begin{equation}
\argmax_{ \substack{ \bm{1}^\mathsf{T}_J  \bm \mu  = 1, \;   \bm{1}^\mathsf{T}_J \bm u  = \myctot  \\  \bm \mu  \geq \bm 0, \; \bm u  \geq \bm 0} }{ \big\{ \bm{c}^\mathsf{T} \log{ \bm P  \bm \mu}  + \log \mathbb{P}(\bm{u}; \bm{\alpha}) \big\} }. 
 \label{eq:preMAP}
\end{equation}
Recalling Eq. \eqref{eq:MLML} we bind the estimate $\hat{\bm u}$ to equal the rescaled value of $\hat{\bm \mu}$, i.e., we impose $\hat{\bm u} = \myctot \hat{\bm \mu}$.  With this simple variable substitution, from Eq. \eqref{eq:preMAP} we obtain the following MAP estimator\footnote{In the derivation we have used the following  invariance property:  a  constant factor $\kappa$ in the argument of the logarithm terms does not influence the solution to the maximization. Formally:
$\argmax_{\bm x, \bm y} \{ \log{\frac{\bm x}{\kappa}} + \bm y\} = \argmax_{\bm x, \bm y} \{ \log{\bm x} - \log{\kappa} +  \bm y\} = \argmax_{\bm x, \bm y} \{ \log{\bm x} +  \bm y \}.$}: 
\begin{equation}
\begin{split}
\hat{\bm u}_\text{MAP} = \argmax_{ \substack{ \bm{1}^\mathsf{T}_J  \bm u   = \myctot \\  \bm u  \geq \bm 0} }{ \big\{ \bm{c}^\mathsf{T}  \log{ \bm P  \bm u}  + \log \mathbb{P}(\bm{u}; \bm{\alpha}) \big\} }  
\end{split}
 \label{eq:MAP00}
\end{equation}
which can be rewritten to finally yield the thesis (note that the term $\bm \alpha$ can be replaced by $\bm a = \bm \alpha/\myctot$ for the same invariance property recalled in the previous footnote).

\subsection{Proof of Proposition~\ref{prop5}}\label{App6}

We start by approximating the (prior) Multinomial distribution with a multivariate normal 
\begin{equation}
\mathbb{P}(\bm{u}; \bm{\alpha}) = \mathcal{M}( \myctot, \bm \alpha) \simeq \mathcal{N}(\bm a, \bm \Sigma_{\bm a})
\label{eq:priorapproximationnormal}
\end{equation}
with mean value $\overline{\bm u}=\bm a =   \myctot \bm \alpha$ and covariance matrix $\Sigma_{\bm a}$ composed of the following elements:
\begin{equation}
\sigma_{mn}^2  = \left\{ \begin{array}{rcl}
 \myctot \, a_{m} \left( 1 - a_{m}\right)   & \mbox{for} & m=n  \\
\myctot \, a_{m}  \,  a_{n}     & \mbox{for} & m \neq n  \\
\end{array}\right. .
   \label{eq:momentsMMN}
  \end{equation}
Plugging the normal approximation Eq. \eqref{eq:priorapproximationnormal} into Eq.  \eqref{eq:MAP00} leads to the approximate  MAP estimator  
\begin{equation}
\hat{\bm u}_\text{A.M.} = \argmax_{ \substack{ \bm{1}^\mathsf{T}_J  \bm u  = \myctot \\  \bm u  \geq \bm 0} }{ \big\{ \bm c^\mathsf{T}  \log{ \bm P  \bm u}  -\frac{1}{2} (\bm u - \bm a)^\mathsf{T} \bm{\Sigma}_{\bm a}^{-1}(\bm u - \bm a) \big\} }  .
 \label{eq:MAP2normal}
\end{equation}

In our application we are considering a very large number of tiles ($J\gg 1$) which implies that the individual probabilities are very small in absolute terms ($\alpha_j \ll 1$). This justifies the further simplification of the covariance matrix:  neglecting the product terms in the off-diagonal elements, $\myctot \alpha_{m}   \alpha_{n} \simeq 0$,  and approximating the diagonal elements  with $ \myctot  \alpha_{m} \left( 1 - \alpha_{m}\right) \simeq \myctot  \alpha_{m} = a_{m}$, the covariance matrix reduces to the diagonal matrix 
$ \bm A \eqdef \mathrm{diag}([a_1 \ \cdots \ a_J ])$
and its inverse to $\bm{\Sigma}_{\bm a}^{-1} \simeq \bm A^{-1} = \mathrm{diag}([\frac{1}{a_1} \ \ldots \ \frac{1}{a_J} ])$. In this way the approximate MAP estimator rewrites 
\begin{equation}
\hat{\bm u}_\text{A.M.} \approx \argmax_{ \substack{ \bm{1}^\mathsf{T}_J  \bm u  = \myctot \\  \bm u  \geq \bm 0} }{ \big\{ \bm c^\mathsf{T}  \log{ \bm P  \bm u}  -\frac{1}{2} (\bm u - \bm a)^\mathsf{T} \bm A^{-1}(\bm u - \bm a) \big\} }  
 \label{eq:MAP2normal2b}
\end{equation}
which finally yields the thesis.

\subsection{Proof of Proposition~\ref{prop6}}\label{App7}

The constrained minimization problem defined by Eq. \eqref{eq:mymle} can be solved by considering the Karush-Kuhn-Tucker conditions:
\begin{equation}
\left\{
\begin{array}{l}
    \nabla_{\bm{u}} \left[ \frac{1}{2} (\bm{u} \!-\! \bm{a})^\mathsf{T} \bm{A}^{-1}(\bm{u} \!-\! \bm{a}) + \bm{\lambda}^\mathsf{T}(\bm{P} \bm{u} - \bm{c}) - \bm{\nu}^\mathsf{T}\bm{u} \right]= \bm{0}  \\
     \bm P \bm u = \bm c\\
     \nu_j \geq 0, \quad \nu_ju_j  = 0, \quad u_j \geq 0, \quad j=1,\ldots,J
\end{array}\right. .
\end{equation}
The calculation of the gradient in the first line yields 
$$
\bm{A}^{-1}(\bm{u} \!-\! \bm{a}) +  \bm{P}^\mathsf{T}\bm{\lambda} - \bm{\nu} = \bm{0}
$$
from which 
\begin{equation}
\bm{u} = \bm{A}\bm{\nu} -  \bm{A} \bm{P}^\mathsf{T}\bm{\lambda} + \bm{a} \label{eq::u}.
\end{equation}
The vector of the Lagrange multipliers $\bm{\lambda}$ corresponding to the equality constraint can be explicitly obtained by exploiting the constraint itself, that is
\begin{equation}
\bm P \bm u = \bm c \Leftrightarrow \bm{P}\bm{A}\bm{\nu} -  \bm{P}\bm{A} \bm{P}^\mathsf{T}\bm{\lambda} + \bm{P}\bm{a} = \bm{c}
\end{equation}
which leads to 
\begin{equation}
\bm{\lambda}=(\bm{P}\bm{A}\bm{P}^\mathsf{T})^{-1} (\bm{P}\bm{A}\bm{\nu}+ \bm{P}\bm{a}-\bm{c}).
\end{equation}
Substituting into Eq. \eqref{eq::u} returns
\begin{align}
\hat{\bm{u}} &= \bm{A}\bm{\nu} - \bm{A} \bm{P}^\mathsf{T} (\bm{P}\bm{A}\bm{P}^\mathsf{T})^{-1} (\bm{P}\bm{A}\bm{\nu}+ \bm{P}\bm{a}-\bm{c})  + \bm{a}. \label{eq:hat_u}
\end{align}

Due to the  constraints $u_j \geq 0$, $\nu_j \geq 0$ and $\nu_j u_j =0$, where $\nu_j>0$ then it must necessarily hold that $\hat{u}_j=0$; conversely, where $\nu_j=0$ the solution is given by Eq. \eqref{eq:hat_u}, which however depends on the overall vector $\bm{\nu}$.
Notice though that the vector $\bm{A}\bm{\nu}$ has zero components for those indices $j$ where $\nu_j=0$, since $\bm{A}$ is a diagonal matrix. The same is not generally true for the vector $\bm{P}\bm{A}\bm{\nu}$, since the further left-multiplication by $\bm{P}$ may recombine the zero elements with the non-zero ones. However, being $\bm{P}$ a very sparse matrix, one can consider a relaxation in which $\nu_j=0$ always produces a corresponding zero entry in $\bm{P}\bm{A}\bm{\nu}$.
This directly yields the vector $\check{\bm u}$ given in Eq. \eqref{eq:check_u}, which however may violate the  constraint on the total mass, i.e., $\bm{1}_J^\mathsf{T}\bm{u}=C$.
Nonetheless, a neighboring solution $\hat{\bm{u}}_\text{DF}$  
can be easily computed by applying to $\check{\bm u}$ the transformation Eq. \eqref{eq:shepp}, that has the effect of  redistributing the mass between non-zero elements of the input vector so as to guarantee the fulfilment of the total mass constraint. 
The outcome of this procedure represents an \emph{ad hoc} estimator that we label DF for ``Data First" and can be written explicitly as:
 \begin{equation}
    \hat{u}_{j,\text{DF}} =  \check{u}_j \, \sum_{i=1}^I  c_i  \frac{ p_{ij}  }{\sum_{k=1}^J  p_{ik}  \, \check{u}_k }  \Longrightarrow \hat{\bm{u}}_\text{DF} = \check{\bm{Q}} \bm{c} 
\end{equation}
with $\check{\bm{Q}}$ given in Eq. \eqref{eq:DF}.

\end{document}